\documentclass[a4paper,11pt]{article}

\usepackage{booktabs} 

\usepackage{amsmath}
\usepackage{amsthm}
\usepackage{amssymb}
\usepackage{amsfonts}
\usepackage{mathrsfs}
\usepackage{color}
\usepackage[utf8]{inputenc}
\usepackage[T1]{fontenc}
\usepackage{enumerate}
\usepackage[english]{babel}
\usepackage[bookmarks=false,colorlinks=true,linkcolor=blue,citecolor=red]{hyperref}
\usepackage[babel]{csquotes} 

\usepackage{graphicx}
\usepackage{latexsym}
\usepackage{epsfig}

\usepackage{thmtools,thm-restate}

\usepackage[ruled]{algorithm2e}

\SetAlFnt{\small}
\SetAlCapFnt{\small}
\SetAlCapNameFnt{\small}
\SetAlCapHSkip{0pt}
\IncMargin{-\parindent}

\theoremstyle{plain}
\newtheorem{theorem}{Theorem}[section]
\newtheorem{lemma}[theorem]{Lemma}
\newtheorem{proposition}[theorem]{Proposition}
\newtheorem{corollary}[theorem]{Corollary}

\theoremstyle{definition}
\newtheorem{definition}[theorem]{Definition}
\newtheorem{example}[theorem]{Example}
\newtheorem{remark}[theorem]{Remark}

\DeclareMathOperator{\Q}{\mathbb Q}

\DeclareMathOperator{\supp}{Supp}
\DeclareMathOperator{\vcsp}{VCSP}
\DeclareMathOperator{\csp}{CSP}
\DeclareMathOperator{\dom}{dom}
\DeclareMathOperator{\feas}{Feas}
\DeclareMathOperator{\fpol}{fPol}

\DeclareMathOperator{\g}{\Gamma}

\DeclareMathOperator{\sub}{sub}
\DeclareMathOperator{\ar}{ar}
\DeclareMathOperator{\QQ}{\mathbb Q\cup\{+\infty\}}

\DeclareMathOperator{\blp}{BLP}
\DeclareMathOperator{\eps}{\boldsymbol\epsilon}

\DeclareMathOperator{\avg}{avg}

\newcommand\xqed[1]{%
  \leavevmode\unskip\penalty9999 \hbox{}\nobreak\hfill
  \quad\hbox{#1}}
\newcommand\demo{\xqed{$\triangle$}}

\begin{document}

\title{Piecewise Linear Valued CSPs\\ Solvable by Linear Programming Relaxation\thanks{This article is the full extended version of the conference paper ``Submodular Functions and Valued Constraint Satisfaction Problems over Infinite Domains" \cite{BodirskyMV18} of the same authors. Although the conclusions of the present work are more general than those contained in the conference version, some of the methods used, the algorithms, and the intermediate results obtained are different and incomparable.}}  \date{}

\author{Manuel Bodirsky \thanks{The authors have received funding from the European Research Council (ERC) under the European	Union's Horizon 2020 research and innovation programme (grant agreement No 681988, CSP-Infinity).}\\{\small  Institut f\"{u}r Algebra}\\{\small Technische Universit\"{a}t  Dresden} \and Marcello Mamino \footnotemark[2]\\{\small Dipartimento di Matematica}\\ {\small Universit\`{a} di Pisa} \and Caterina Viola\footnotemark[2] \thanks{This author has been supported by DFG Graduiertenkolleg 1763 (QuantLA).}\\{\small Institut f\"{u}r Algebra}\\ {\small Technische Universit\"{a}t  Dresden}}



\maketitle

\begin{abstract}
Valued constraint satisfaction problems (VCSPs) are a large class of combinatorial optimisation problems. The computational complexity of VCSPs depends on the set of allowed
cost functions in the input. Recently, the computational complexity of
all VCSPs for finite sets of cost functions over finite domains has been
classified. Many natural optimisation problems, however, cannot be formulated as VCSPs over a finite domain. 
We initiate the systematic investigation of infinite-domain VCSPs by studying the complexity of VCSPs for piecewise linear homogeneous cost functions. Such VCSPs can be solved in polynomial time if the cost functions
are improved by fully symmetric fractional operations of all arities. We
show this by reducing the problem to a finite-domain VCSP which can
be solved using the basic linear program relaxation. It follows that VCSPs for submodular PLH cost functions can be solved
in polynomial time; in fact, we show that submodular PLH functions
form a maximally tractable class of PLH cost functions.
\end{abstract}

\section{Introduction}
In a \emph{valued constraint satisfaction problem (VCSP)} 
we are given a finite set of variables, a finite
set of cost functions that depend on these variables, and a cost $u$; the task is to find values for the variables such that the sum of the cost functions is at most $u$. 
By restricting the set of possible cost functions in the input, a great variety of computational optimisation problems can be modelled as a valued constraint satisfaction problem. By allowing the cost functions to evaluate to $+\infty$, we can even model `crisp' (i.e., hard) constraints on the variable assignments, 
and hence the class of (classical) constraint satisfaction problems (CSPs) is a subclass of the class of all VCSPs. 

If the domain is finite, the computational complexity of VCSPs has recently been classified for all sets of cost functions, assuming the Feder-Vardi conjecture for classical CSPs~\cite{KolmogorovThapperZivny,GenVCSP15,KozikOchremiak15}. Even more recently, two solutions to the Feder-Vardi conjecture have been announced~\cite{ZhukFVConjecture,BulatovFVConjecture}. These fascinating achievements settle the complexity of VCSPs over finite domains.

Several outstanding combinatorial optimisation problems cannot be formulated as VCSPs over a finite domain, but they can be formulated
as VCSPs over the domain ${\mathbb Q}$, the set of rational numbers. 
One example is the famous linear programming problem, where the task is to optimise a linear function subject to linear inequalities. This can be modelled as a VCSP by allowing unary linear cost functions and cost functions of higher arity to express the crisp linear inequalities. Another example is 
the minimisation problem for sums of piecewise linear 
convex cost functions (see, e.g., \cite{BoydVandenberghe}). Both of these problems can be solved in polynomial time, e.g.~by the ellipsoid method (see, e.g., \cite{GroetschelLovaszSchrijver}). 

Despite the great interest in such concrete VCSPs over the rational numbers in the literature, VCSPs over infinite domains have not yet been studied systematically. 
In order to obtain general results we
need to restrict
the class of cost functions that we investigate,
because without any restriction it is already hopeless to classify the complexity of infinite-domain CSPs (any language over a finite alphabet is polynomial-time Turing equivalent to an infinite domain CSP~\cite{BodirskyGrohe}). One restriction that captures a variety 
of optimisation problems of theoretical and practical interest is the class of all \emph{piecewise linear homogeneous} cost functions over ${\mathbb Q}$,
defined below. We first illustrate by an example the type of cost functions that we want to 
capture in our framework.

\begin{example}{(\textsc{Least correlation clustering with partial information})}\label{expl:intro1}
	We consider the following problem:
we are given a graph on $n$ vertices such that the set of edges $E$ is partitioned in two classes, $E_{-}$ and $E_{+}$. An edge $(x,y)\in E$ is either in $E_{+}$ or in $E_{-}$ depending on whether $x$ and $y$ have been deemed to be  similar or different. The goal is to (decide whether it is possible to) produce a partition of the vertices, namely a clustering, that agrees with the edge partition on at least $l$ edges, where $l$ is a given (rational) number between $0$ and $\lvert E \rvert$. That is, we want a clustering that bounds the number of disagreements, i.e., the number of edges from $E_{+}$ between clusters plus the number of edges from $E_{-}$ inside clusters. This problem can be seen as an instance of a VCSP with variables $x_1,\ldots, x_n$, with objective function \[\phi(x_1,\ldots,x_n)=\sum_{(x_i,x_j)\in E_{+}}f_1(x_i,x_j)+\sum_{(x_i,x_j)\in E_{-}}f_2(x_i,x_j),\] and with threshold $u:=\lvert E \rvert-l$, where the cost functions $f_1,f_2 \colon \Q^2 \to \Q$ are defined by \[f_1(x_i,x_j)=\begin{cases}
0 & \text{if } x_i=x_j\\
1 & \text{otherwise}
\end{cases}, \; \text{and} \quad f_2(x_i,x_j)=\begin{cases}
1 & \text{if } x_i=x_j\\
0 & \text{otherwise}
\end{cases}.\]
Observe that this problem cannot be modelled as a VCSP over a finite domain, as we do not want to bound the possible number of clusters and we want want to allow graphs with any finite number of vertices as input. 
The least correlation clustering with partial information that we defined above is a generalisation of the well-known min-correlation clustering problem (see \cite{CorrelationClustering}). As the min-correlation clustering problem is known to be NP-hard, so is our problem. 
\demo
	\end{example}

A partial function 
$f \colon {\mathbb Q}^n \to {\mathbb Q}$ is called
\emph{piecewise linear homogeneous (PLH)} 
if it is first-order definable over the structure ${\mathfrak L} := ({\mathbb Q};<,1,(c\cdot)_{c \in {\mathbb Q}})$; being undefined at 
$(x_1,\dots,x_n) \in {\mathbb Q}^n$ is interpreted as $f(x_1,\dots$, $x_n) = +\infty$. 
The structure $\mathfrak L$ has quantifier elimination (see Section~\ref{sect:piecewise-hom}) and hence there are finitely many regions
such that $f$ is a homogeneous linear polynomial in each region; this is the motivation for the name \emph{piecewise linear homogeneous}. 
The cost function from Example~\ref{expl:intro1} is PLH.

In this article we present a sampling technique to solve VCSPs for
PLH cost functions in polynomial time.
The technique consists of a polynomial-time
many-one reduction from the VCSP for a finite set of PLH cost functions to the VCSP for a \emph{sample},
i.e., the same set of cost functions interpreted over a suitable finite domain. We present this technique in two steps: first, in Section \ref{sect:csp-tract} we present the reduction for the {\it feasibility problem}, i.e., the problem of deciding whether there exists an assignment of values for the variables with finite cost; second, in Section \ref{sect:tract} we extend this method to solve the {\it optimisation problem}, i.e., the task of deciding whether there exists an assignment of values for the variables with cost at most equal to a given threshold.  
In Section \ref{sect:totsym} we present a sufficient condition under which infinite-domain VCSPs that admit
an efficient sampling algorithm can be solved in polynomial time using linear
programming relaxation. The condition is given in terms of algebraic
properties of the cost functions, namely the existence of \emph{fully symmetric fractional polymorphisms of all arities}.
We combine the results mentioned above to solve the VCSP for sets of PLH cost functions having such fractional polymorphisms.
 In particular, we apply the combination to show containment in P of the VCSP
for {\it submodular}  PLH cost functions,  for {\it convex} PLH cost functions, and for  {\it componentwise increasing} PLH cost functions. 
 
Submodular cost functions naturally appear in several scientific fields such as, for example, economics, game theory, machine learning, and computer vision, and play a key role in operational research and combinatorial optimisation (see, e.g.,~\cite{Fujushige}).   Submodularity also plays an important role for 
the computational complexity of VCSPs over finite domains, and guided the research on VCSPs for some time (see, e.g., \cite{cohen2006complexity,JKT11}), 
even though this might no longer be visible in the final classification obtained in~\cite{KolmogorovThapperZivny,GenVCSP15,KozikOchremiak15}. 
The polynomial-time tractability of VCSPs for submodular PLH cost functions is the main result of an extended abstract  \cite{BodirskyMV18} which announced some of the results presented here: we first provided a polynomial-time many one reduction to a  finite-domain VCSP with costs in the ring $\Q^\star:=\Q((\epsilon))$ of formal Laurent power series in one indeterminate. Then we used a fully combinatorial polynomial-time algorithm to solve the VCSP for the $\Q^\star$-valued submodular sample.

 In the present article, starting from the same polynomial-time many-one reduction, we provide a polynomial-time many-one reduction to a finite-domain VCSP with values in $\Q$   and then use linear programming relaxation to solve the finite-domain problem. The approach we present here is more general: it can be applied to all VCSPs for PLH cost functions that have fully symmetric fractional polymorphisms of all arities and it does not require the existence of fully combinatorial algorithms solving the finite-domain VCSP computed by our reduction. The concept of fractional polymorphisms comes from universal algebra and it has been fruitfully used to characterise the computational complexity of finite-domain VCSPs (see, e.g.,  \cite{KolmogorovThapperZivny}, \cite{ThapperZivny13},  \cite{KozikOchremiak15}, \cite{GenVCSP15}).
 At the end of Section~\ref{sect:submodular}, we also show that submodularity defines
a \emph{maximal tractable class} of PLH valued structures: informally, adding any 
cost function that is not submodular leads to an NP-hard VCSP. 
 Section~\ref{sect:conclusion} closes with some problems and challenges for future research.


\section{Valued Constraint Satisfaction Problems}
\label{sect:vcsps}
A \emph{valued structure $\Gamma$ 
(over $D$)}  consists of 
\begin{itemize}
\item a signature $\tau$ 
consisting of function symbols $f$, each equipped with an arity $\ar(f)$, 
\item a set $D = \dom(\Gamma)$ (the \emph{domain}), 
\item for each $f \in \tau$ a 
\emph{cost function}, i.e.,  a function $f^{\Gamma} \colon D^{\ar(f)} \to {\mathbb Q} \cup \{+\infty\}$. 
\end{itemize}
Here, $+\infty$ is an extra element with the expected properties that for all $c \in {\mathbb Q} \cup \{+\infty\}$
\begin{align*}
(+\infty) + c & = c + (+\infty) = +\infty \\
\text{  and } c & < +\infty \text{ iff } c \in {\mathbb Q}.
\end{align*} 
Let $\Gamma$ be a valued structure with a finite signature $\tau$. 
The \emph{valued constraint satisfaction problem  for $\Gamma$}, denoted by $\vcsp(\g)$, is the following computational problem. 

\begin{definition}\label{vcspdef}
An \emph{instance} $I$ of $\vcsp(\g)$ 
consists of 
\begin{itemize}
\item a finite set of variables $V_I$, 
\item an expression $\phi_I$ of the form
\[\sum_{i=1}^{m} f_i(x^i_1,\ldots, x^i_{\ar(f_i)})\]
where $f_1,\dots,f_m \in \tau$ and all the $x^i_j$ are variables from $V_I$, and
\item a value $u_I \in {\mathbb Q}$. 
\end{itemize} 
The task is to decide whether there exists an assignment $\alpha \colon V_I \to \dom(\Gamma)$, whose \emph{cost}, defined as
\[\sum_{i=1}^{m} f^\Gamma_i(\alpha(x^i_1),\ldots, \alpha(x^i_{\ar(f_i)}))\]
is finite, and if so, whether there is one whose cost 
is at most $u_I$.
\end{definition}

 Note that, given a valued structure $\g$, if the signature $\tau$ of $\Gamma$
is finite, it is inessential for the computational complexity
of $\vcsp(\g)$ how the function symbols in $\phi_I$ are represented.

The function from $\phi_I^{\g}\colon\dom(\g)^{\lvert V_I\rvert} \to\QQ$  described by the expression $\phi_I$ 
is also called the \emph{objective function}.  
The problem of deciding whether there exists an  assignment $\alpha\colon V \to \dom(\g)$ with finite cost
is called the \emph{feasibility problem}, which can
also be modelled as a (classical) constraint satisfaction problem (cf.~Section \ref{csp-sect}). The choice of defining the VCSP as a decision problem and not as an optimisation problem as it is common for VCSPs over finite domains is motivated by two major issues that do not occur in the finite-domain case: first, in the infinite-domain setting one  needs to capture the difference between a proper minimum and an infimum value that the cost of the assignment can be  arbitrarily close  to but never reach; second, our definition allows to model the case in which the infimum is $-\infty$, i.e., when there are assignments for the variables of arbitrarily small cost.

By a \emph{finite-valued structure} we refer to a valued structure whose cost functions are  finite-valued, i.e., they assume a finite value on every point with rational coordinates. The VCSP for a finite-valued structure is merely an optimisation problem and  it does not involve any associated feasibility problem, i.e., the problem asking wheather there is an assignment with finite cost.

VCSPs have been studied intensively
when $D = \dom(\Gamma)$ is finite,
and as mentioned in the introduction, in this case a complete classification of the computational complexity of 
$\vcsp(\Gamma)$ has been obtained recently. 
However, many well-known optimisation problems can only be formulated when we allow
infinite domains $D$.

\begin{example}\label{expl:LP}
Let $\Gamma$ be the valued structure
with domain $D := {\mathbb Q}$ and
the signature $\tau = \{R_+,R_1,\leq,\text{id}\}$ where
\begin{itemize}
\item $R_+$ is ternary, and
\[R^\Gamma_+(x,y,z) = \begin{cases} 0 
& \text{ if } x+y=z \\
+\infty & \text{ otherwise; } \end{cases}\]
\item $R_1$ is unary and 
\[R_1^{\Gamma}(x) := \begin{cases} 0 
& \text{ if } x=1 \\
+\infty & \text{ otherwise; } \end{cases}\]
\item $\leq$ is binary and 
\[\leq^{\Gamma}(x,y)  := \begin{cases} 0 
& \text{ if } x \leq y \\
+\infty & \text{ otherwise; } \end{cases}\]
\item $\text{id}$ is unary and 
\[\text{id}^{\Gamma}(x)  := x.\]
\end{itemize}
Then instances of $\vcsp(\Gamma)$ can be viewed as instances of the linear program feasibility problem, i.e., the problem of deciding whether there exists a solution satisfying finitely
	many linear inequality constraints.  \demo
\end{example}

We give another example to illustrate the 
flexibility of the VCSP framework for formulating optimisation problems; the valued structure in this example contains non-convex cost functions, but, as we will see later, can be solved in polynomial time. 

\begin{example}\label{expl:intro2}
Let $\Gamma$ be the valued structure
with signature $\tau = \{g_1,g_2,g_3\}$ and
the cost functions 
\begin{itemize}
\item $g^\Gamma_1 \colon {\mathbb Q} \to {\mathbb Q}$ defined by $g_1(x) = -x$,
\item $g^\Gamma_2 \colon {\mathbb Q}^2 \to \mathbb Q$ defined by $g_2(x,y) := \min(x,-y)$,
and 
\item $g^\Gamma_3 \colon {\mathbb Q}^3 \to \mathbb Q$ defined by $g_3(x,y,z) := \max(x,y,z)$. 
\end{itemize}
Two examples of instances of VCSP$(\Gamma)$ are 
\begin{align} 
& g_1(x) + g_1(y) + g_1(z) + g_2(x,y) \nonumber \\
+ &  
g_3(x,y,z) + g_3(x,x,x) + g_3(x,x,x)  \label{eq:instance1} \\
\text{ and } \quad & 
g_1(x) + g_1(y) + g_1(z) \nonumber \\
+ & g_3(x,y,z) + g_3(x,x,y) + g_3(y,z,z)
\label{eq:instance2}
\end{align}
We can make the cost function described by the expression in $(\ref{eq:instance1})$ arbitrarily small by fixing $x$ to $0$ and choosing $y$ and $z$ sufficiently large. 
On the other hand, the minimum for the cost function in $(\ref{eq:instance2})$ is $0$, obtained by setting $x,y,z$ to $0$. 
Note that $g_1$ and $g_3$ are convex functions, but $g_2$ is not. \demo
\end{example}

\section{Cost Functions over the Rationals}
\label{sect:cost}
The class of all valued structures over arbitrary infinite domains is too large to allow for complete complexity classifications, so we have to restrict our focus to subclasses. In this section we describe natural 
and large classes of cost functions over the domain $D = {\mathbb Q}$, the rational numbers. These classes
are most naturally introduced using first-order definability, so we briefly fix the necessary logic concepts. 

\subsection{Logic Preliminaries}
We fix some standard logic terminology; see, e.g.,~\cite{Hodges}. 
A \emph{signature} is a set $\tau$ of function and relation symbols. Each function symbol $f$ and each relation symbol $R$ is equipped with an arity $\ar(f)$, $\ar(R) \in {\mathbb N}$.  
A \emph{$\tau$-structure} $\mathfrak A$ consists of \begin{itemize}
\item a set $A = \dom(\mathfrak A)$, called the \emph{domain} of $\mathfrak A$, whose elements are called the \emph{elements} of the $\tau$-structure; 
\item  a relation $R^{\mathfrak A} \subseteq A^{\ar(R)}$ for each relation symbol $R \in \tau$; 
\item  a function $f^{{\mathfrak A}} \colon A^{\ar(f)} \to A$ for each function symbol $f \in \tau$. 
\end{itemize}
Function symbols of arity $0$ are allowed and
are called constant symbols. 
We give two examples of structures that play an important role in this article. 

\begin{example}\label{expl:S}
Let $\mathfrak S$ be the structure
with domain ${\mathbb Q}$ 
and signature $\sigma:=\{+,1,\leq\}$
where 
\begin{itemize}
\item $+$ is a binary function symbol that 
denotes the usual addition over ${\mathbb Q}$, \item $1$ is a constant symbol that denotes $1 \in {\mathbb Q}$, and 
\item $\leq$ is a binary relation symbol that denotes the usual linear order of the rationals. \demo
\end{itemize}
\end{example}

\begin{example}\label{expl:L}
Let $\mathfrak L$ be the structure with domain $\Q$ and
(countably infinite) signature 
$\tau_0 := \{<,1\}\cup \{c\cdot\}_{c \in \Q}$ 
where 
\begin{itemize}
\item $<$ is a relation symbol of arity $2$
and $<^{\mathfrak L}$ is the strict linear order of $\Q$,
\item $1$ is a constant symbol and $1^{\mathfrak L} := 1 \in \Q$, and 
\item $c \cdot$ is a unary function symbol for every $c \in {\mathbb Q}$ such that $(c \cdot)^{\mathfrak L}$ is the function $x \mapsto cx$ (multiplication by $c$).\demo
\end{itemize} 
\end{example}

A \emph{relational structure} (or structure with \emph{relational signature}) is a structure whose signature contains only relation symbols. A structure (or a valued structure) is called \emph{finite} if its domain is finite. 

\subsection{Constraint Satisfaction Problems}\label{csp-sect}

Let $\g$ be a  valued  structure with signature $\tau$.
The question whether an instance of $\vcsp(\Gamma)$ is \emph{feasible},
that is, has an assignment with finite cost, can be viewed as a (classical) \emph{constraint satisfaction
problem}. Formally, the constraint 
satisfaction problem for a structure $\mathfrak A$ with a finite
relational signature $\tau$ is the following computational problem,
denoted by $\csp(\mathfrak A)$: 
\begin{itemize}
	\item the input is a finite conjunction $\psi$ of atomic $\tau$-formulas,
	and 
	\item the question is whether 
	$\psi$ is satisfiable in $\mathfrak A$. 
\end{itemize}

We can associate to $\Gamma$ the following
relational structure $\feas(\Gamma)$: for every function symbol $f$ of arity
$n$ from $\tau$ the 
signature of $\feas(\Gamma)$ contains a 
relation symbol $R_f$ of arity $n$ such that 
$R_f^{\feas(\Gamma)} = \dom(f^{\g})$. 
Every polynomial-time algorithm for
$\vcsp(\g)$ in particular has to solve $\csp(\feas(\g))$.
In fact, an instance $\phi$ of $\vcsp(\g)$ can be translated into an
instance $\psi$ of $\csp(\feas(\g))$ by replacing 
subexpressions of the form 
$f(x_1,\dots,x_n)$ in $\phi$ by 
$R_f(x_1,\dots,x_n)$
and by replacing $+$ by $\wedge$. 
It is easy to see that $\phi$ is a feasible instance of $\vcsp(\Gamma)$ if,
and only if, $\psi$ is satisfiable in $\feas(\g)$.

\subsection{Quantifier Elimination} 

We adopt the usual definition of first-order logic. 
A formula is \emph{atomic} if it does not contain logical symbols (connectives or quantifiers). By convention, we have two special atomic formulas,
$\top$ and~$\bot$, to denote truth and falsity.

Let $\tau$ be a signature. We say that a $\tau$-structure $\mathfrak A$ has \emph{quantifier elimination} if every first-order $\tau$-formula is equivalent to a quantifier-free $\tau$-formula over $\mathfrak A$.

\begin{theorem}[\cite{FerranteRackoff}, Section 3, Theorem 1]
The structure $\mathfrak S$ from Example~\ref{expl:S} has quantifier elimination. 
\end{theorem}

\begin{theorem}\label{thm:qe}
The structure $\mathfrak L$ from Example~\ref{expl:L} has quantifier elimination. 
\end{theorem}

To prove Theorem \ref{thm:qe} it suffices to prove the following lemma.

\begin{lemma}\label{qel}
For every quantifier-free $\tau_0$-formula $\varphi$ there exists a quantifier-free $\tau_0$-formula $\psi$ such that $\exists x. \varphi$ is equivalent to $\psi$ over $\mathfrak L$. 
\end{lemma}

\begin{proof}
Observe that every atomic $\tau_0$-formula $\phi$
has at most two variables:
\begin{itemize}
\item if $\phi$ has no variables,
then it is equivalent to $\top$ or $\perp$,
\item if $\phi$ has only one variable, say $x$,
then it is equivalent to $c {\cdot} x \, \sigma \, d {\cdot} 1$ or to $d {\cdot} 1 \, \sigma \, c {\cdot} x$ for $\sigma \in \{<,=\}$ and $c,d \in {\mathbb Q}$. 
Moreover, if $c = 0$ then $\phi$ is equivalent to a formula without variables, and otherwise $\phi$ is equivalent to $x \, \sigma \, \frac{d}{c}{\cdot} 1$
or to $ \frac{d}{c}{\cdot}1 \, \sigma \, x$ for $\sigma \in \{<,=\}$, which we abbreviate by the more common $x < \frac{d}{c}$, $x = \frac{d}{c}$, and 
$\frac{d}{c}<x$, respectively. 
\item if $\phi$ has two variables, say $x$ and $y$,
then $\phi$ is equivalent to $c {\cdot} x \, \sigma \, d {\cdot} y$ or $c {\cdot} x \, \sigma \, d {\cdot} y$ for $\sigma \in \{<,=\}$. Moreover, if $c = 0$ or $d = 0$ then the formula $\phi$ is equivalent to a formula with at most one variable, and otherwise $\phi$ is equivalent to $x \, \sigma \, \frac{d}{c}{\cdot}y$ or to $\frac{d}{c}{\cdot} y \, \sigma \, x$.
\end{itemize} 

We define $\psi$ in five steps. \begin{enumerate}
\item Rewrite $\varphi$, using De Morgan's laws, in such a way that all the negations are applied to atomic formulas. 
\item Replace \begin{itemize} 
\item $\neg (s=t)$ by $s<t\vee t<s$, and 
\item $\neg (s<t)$ by $t<s \vee s=t$,
\end{itemize} where $s$ and $t$ are $\tau_0$-terms.
\item Write $\varphi$ in disjunctive normal form in such a way that each of the clauses is a conjunction of non-negated atomic $\tau_0$-formulas (this can be done by distributivity).
\item Observe that $\exists x \bigvee_i \bigwedge_j \chi_{i,j}$, where the $\chi_{i,j}$ are atomic $\tau_0$-formulas, is equivalent to $\bigvee_i \exists x \bigwedge_j \chi_{i,j}$. Therefore, it is sufficient to prove the lemma for $\varphi=\bigwedge_j \chi_j$ where the $\chi_j$ are  atomic $\tau_0$-formulas. As explained above, we can assume without loss of generality that the $\chi_j$ are of the form $\top$, $\perp$, $x \, \sigma \, c$,  $c \, \sigma \, x$, or $x \, \sigma \, cy$, for $c \in {\mathbb Q}$ and $\sigma \in \{<,=\}$. 
If $\chi_j$ equals $\perp$, then $\varphi$
is equivalent to $\perp$ and there is nothing to be shown. If $\chi_j$ equals $\top$ then it can simply
be removed from $\varphi$. If
$\chi_j$ equals $x = c$ or $x=cy$ then replace every occurrence of $x$ by $c \cdot 1$ or by $c\cdot y$, respectively. 
Then $\varphi$ does not contain the variable $x$ anymore and thus $\exists x. \varphi$ is equivalent to $\varphi$. 
\item We are left with the case that all atomic $\tau_0$-formulas involving $x$ are (strict) inequalities, that is, $\varphi=\bigwedge_i \chi_i \wedge \bigwedge_i \chi'_i \wedge \bigwedge_i \chi''_l$, where 
\begin{itemize}
\item the $\chi_i$ are atomic formulas not containing $x$,
\item the $\chi'_i$ are atomic formulas of the form $x>u_i$,
\item the $\chi''_i$ are atomic formulas of the form $x<v_i$.
\end{itemize}
Then $\exists x. \varphi$ is equivalent to $\bigwedge_i \chi_i \wedge \bigwedge_{i,j}(u_i<v_j)$.
\end{enumerate}
Each step of this procedure preserves the satisfying assignments for $\varphi$ and the resulting formula is in the required form; 
this is obvious for all but the last step, 
and for the last step follows from the correctness of Fourier-Motzkin elimination for systems of linear inequalities (see, e.g., \cite{Schrijver}, Section 12.2). Therefore, the procedure is correct.
\end{proof}

\begin{proof}[(of Theorem \ref{thm:qe})]
Let $\varphi$ be a $\tau_0$-formula. We prove that it is equivalent to a quantifier-free $\tau_0$-formula by induction on the number $n$ of quantifiers of $\varphi$. For $n=1$ we have two cases:
\begin{itemize}
\item If $\varphi$ is of the form $\exists x. \varphi'$ (with $\varphi'$ quantifier-free) then, by Lemma \ref{qel}, it is equivalent to a quantifier-free $\tau_0$-formula $\psi$.
\item If $\varphi$ is of the form $\forall x. \varphi'$ (with $\varphi'$ quantifier-free), then it is equivalent to $\neg \exists x. \neg \varphi'$. By Lemma \ref{qel}, $\exists x. \neg \varphi'$ is equivalent to a quantifier-free $\tau_0$-formula $\psi$. Therefore, $\varphi$ is equivalent to the quantifier-free $\tau_0$-formula $\neg \psi$.
\end{itemize}
Now suppose that $\varphi$ is of the form
$Q_1 x_1 Q_2x_2 \cdots Q_nx_n.\varphi'$ 
for $n \geq 2$ and $Q_1,\dots,Q_n \in \{\forall,\exists\}$, and suppose that the statement 
is true for $\tau_0$-formulas with at most $n-1$ quantifiers. In particular, 
$Q_2x_2 \cdots Q_nx_n.\varphi'$ is equivalent to a quantifier -free $\tau_0$-formula $\psi$. Therefore, $\varphi$ is equivalent to $Q_1x_1.\psi$, that is, a $\tau_0$-formula with one quantifier that is equivalent to a quantifier-free $\tau_0$-formula, again by the inductive hypothesis.
\end{proof}

\subsection{Piecewise Linear Homogeneous Functions}
\label{sect:piecewise-hom}
A \emph{partial function} of arity $n \in {\mathbb N}$ over a set $A$ is a function
\[f \colon \dom(f) \to A \text{ for some } \dom(f) \subseteq A^{n} \; .\]
Let $\mathfrak A$ be a $\tau$-structure. 
A partial function over $A$ is called \emph{first-order definable over $\mathfrak A$} if 
there exists a first-order $\tau$-formula $\phi(x_0,x_1,\dots,x_n)$ such that for all $a_1,\dots,a_n \in A$ 
\begin{itemize}
\item if $(a_1,\dots,a_n) \in \dom(f)$ then $\mathfrak A \models \phi(a_0,a_1,\dots,a_n)$
if, and only if,\\ $a_0 = f(a_1,\dots,a_n)$, and
\item if $f(a_1,\dots,a_n) \notin \dom(f)$ then there is no $a_0 \in A$ such that\\ $\mathfrak A \models \phi(a_0,a_1,\dots,a_n)$. 
\end{itemize}

In the following, 
we consider \emph{cost functions
over ${\mathbb Q}$}, which will be  
functions from ${\mathbb Q}^n \to \Q \cup \{+\infty\}$.
It is sometimes convenient to view a cost function 
as a partial function over ${\mathbb Q}$. 
 If $t \in \Q^{\ar(f)} \setminus \dom(f)$ we interpret this as $f(t) = +\infty$.

\begin{definition}
A cost function $f \colon {\mathbb Q}^n \rightarrow \Q \cup \{+\infty\}$ (viewed as a partial function)
is called \begin{itemize}
\item \emph{piecewise linear (PL)} if it is 
first-order definable over $\mathfrak S$, piecewise linear functions are sometimes called \emph{semilinear} functions;
\item \emph{piecewise linear homogeneous (PLH)} if it is 
first-order definable over $\mathfrak L$ (viewed as a partial function). 
\end{itemize}
A valued structure $\Gamma$ is called \emph{piecewise linear} (\emph{piecewise linear homogeneous}) if every cost function in $\Gamma$ is PL (or PLH, respectively). 
\end{definition}

\begin{definition}A relational structure  $\mathfrak A$ with domain $\Q$ and relational signature $\tau$  is called \emph{piecewise linear homogeneous (PLH)} if, for all $R\in\tau$, the interpretation $R^{\mathfrak A}$ is PLH. 

\end{definition}
Clearly, if $\g$  is a PLH valued structure then $\feas(\Gamma)$ is a  PLH relational structure.
Every PLH
cost function is also PL, since
all functions of the structure $\mathfrak L$ are clearly first-order definable in $\mathfrak S$. 
The cost functions in the valued structure from Example~\ref{expl:intro2} are PLH.
The cost functions in the valued structure from Example~\ref{expl:LP} are PL, but not PLH.

We would like to point out that already the class of 
PLH cost functions is very large.  In particular, it can be seen as a generalisation of the class of all cost functions over a finite domain. Indeed,
every VCSP for a valued structure over a finite domain is also a VCSP for a valued structure that is PLH. To see this, suppose that $f \colon D^d \to \Q \cup \{+\infty\}$ is such a cost function,
identifying $D$ with a subset of $\Q$ in an arbitrary way. Then the function $f' \colon \Q^d \to \Q \cup \{+\infty\}$
defined by $f'(x_1,\dots,x_n) := f(x_1,\dots,x_n)$
if $x_1,\dots,x_n \in D$, and
$f'(x_1,\dots,x_n) = +\infty$ otherwise, 
is PLH.

In the following, we prove that the VCSP for a PLH valued structure with finite signature is polynomial-time many-one equivalent to the VCSP for a valued structure over a suitable finite domain. This will be done in two steps: firstly we will show the reduction for the feasibility problem, i.e., we will prove such a reduction for PLH relational  structure; secondly we will extend this method to solve the {\it optimisation problem}, i.e., to find a solution of cost at most equal to the threshold given in the instance.

\section{Efficiently Sampling a PLH Relational Structure}	\label{sect:csp-tract}
\def\mfa{{\mathfrak A}}
Throughout this section, we fix a  PLH relational structure $\mathfrak A$ with finite signature $\tau$.
We present an efficient sampling algorithm for $\mathfrak A$.  Before defining the notion of sampling algorithm, we recall a well-known concept from universal algebra.

\begin{definition}
	Let $\mathfrak A$ and $\mathfrak B$ be (relational) structures with the same signature $\tau$ with domain $A$ and $B$ respectively. A \emph{homomorphism}  from $\mathfrak A$ to $\mathfrak B$ is a function $h \colon A\to B$ such that for every relation symbol $R \in \tau$ and tuple $a \in D^{\ar(R)}$ \[ R^{\mathfrak A}(a) \text{ implies } R^{\mathfrak B}\left(g(a)\right),\]
	where the function $g$ is applied componentwise. We say that $\mathfrak A$ is \emph{homomorphic} to $\mathfrak B$ and write $\mathfrak A \to \mathfrak B$ to indicate the existence of a  homomorphism from $\mathfrak A$ to $\mathfrak B$.
\end{definition}

We now formally introduce the notion of a sampling algorithm for a relational structure.

\begin{definition}
Let $\mathfrak C$ be a structure	with finite relational signature $\tau$.
A \emph{sampling algorithm for $\mathfrak C$}  takes as input a
positive integer $d$ and computes a finite-domain structure $\mathfrak D$  homomorphic to $\mathfrak C$
such that every finite conjunction of atomic $\tau$-formulas having at
most $d$ distinct free variables is
satisfiable in $\mathfrak C$ if, and only if, it is satisfiable in
$\mathfrak D$.
A sampling algorithm is called \emph{efficient} if its running time is
bounded by a polynomial in $d$.
We refer to the output of a sampling algorithm by calling it the \emph{sample}.
 \end{definition}

The definition above is a slight re-formulation of Definition~2.2
in~\cite{BodMacpheTha},
and it is easily seen to give the same results using the same proofs.
We decided to bound the number of variables
instead of the size of the conjunction of atomic~$\tau$-formulas because
this is more natural in our context. These two quantities are polynomially
related by the assumption that the signature~$\tau$ is finite.

We give a formal definition of
the {\it numerical data} in~$\mfa$, we will need it later on. By quantifier
elimination (Theorem~\ref{thm:qe}), each of the finitely many relations $R^\mfa$ for~$R\in\tau$ has a
a quantifier-free $\tau_0$-formula~$\phi_R$ over $\mathfrak L$. As
in the proof of Theorem~\ref{thm:qe}, we can assume that all formulas~$\phi_R$
are positive (namely contain no negations).
From now on, we will fix one such representation. Let
\def\atm#1{{\text{At}(#1)}}$\atm{\phi_R}$ denote the set of atomic subformulas of~$\phi_R$. Each atomic
$\tau_0$-formula is of the form $t_1 {<\atop =} t_2$, where $t_1$
and~$t_2$ are terms. We call an atomic formula \emph{trivial} if
it is  equivalent to~$\bot$ or~$\top$, and \emph{non-trivial} otherwise.
As in the proof of Theorem~\ref{thm:qe},  we make the assumption that atomic formulas
are of the form~$\bot$ or~$\top$ if they are trivial, and otherwise of the form  either $c_1\cdot 1 {<\atop =} x_i$, or $x_i {<\atop =} c_2\cdot 1$, or $c_1\cdot x_i {<\atop =} c_2\cdot x_j$ with constants~$c_1$ and~$c_2$  not both negative and where function symbols $c_i\cdot$ are never composed. This assumption can be made without loss of generality (again, see the proof of Theorem~\ref{thm:qe}).

Given a set of non-trivial atomic formulas~$\Phi$, we define
\[ H(\Phi) = \left\{ \frac{c_1}{c_2} \;\middle|\; t_1=c_1\cdot x_i, \; t_2
= c_2\cdot x_j,\;  \text{for some $t_1 {<\atop =} t_2$ in $\Phi$} \right\}
\]
\begin{align*}
K(\Phi) = &\left\{ \frac{c_2}{c_1} \;\middle|\; t_1=c_1\cdot x_i, \; t_2 =
c_2\cdot 1,\;  \text{for some $t_1 {<\atop =} t_2$ in $\Phi$} \right\} \\
\cup
&\left\{ \frac{c_1}{c_2} \;\middle|\; t_1=c_1\cdot 1, \; t_2 = c_2\cdot
x_j,\;  \text{for some $t_1 {<\atop =} t_2$ in $\Phi$} \right\}
\end{align*}

The efficient sampling algorithm for $\mfa$ works in two steps. First, the problem $\csp(\mfa)$ is transferred to the equivalent $\csp$ for a suitable structure $\mfa^\star$ that is an  extension of an expansion (or an expansion of an extension) of $\mfa$, and second we provide an efficient sampling algorithm for $\mfa^\star$, which is also an efficient sampling algorithm for $\mfa$.


\def\eps{{\boldsymbol\epsilon}}
\begin{definition}
The ordered $\Q$-vector space~$\Q^\star$ is defined as
\[ \Q^\star = \{ x + y{\boldsymbol\epsilon} \;\mid\; x,y\in\Q\} \]
where $\boldsymbol\epsilon$ is merely a formal device, namely
$x+y{\boldsymbol\epsilon}$ represents the pair~$(x,y)$.
We define addition and multiplication by a scalar componentwise
\begin{align*}
(x_1 + y_1\eps)\,+\,(x_2+y_2\eps) &= (x_1+x_2) + (y_1+y_2)\eps \\
c\cdot(x + y\eps) &= (cx) + (cy) \eps \text{.}
\end{align*}
Clearly, $\Q$ is embedded in~$\Q^\star$, by mapping every rational number $x$ into $x+0\eps$.
The order is induced by~$\Q$ extended with $0 < \boldsymbol\epsilon \ll 1$,
namely the lexicographical order of the components $x$ and~$y$
\[
(x_1 + y_1\eps)\,<\,(x_2+y_2\eps)\quad\text{iff}\quad
\begin{cases}
x_1 < x_2 \quad \text{or} \\
x_1 = x_2 \;\wedge\; y_1 < y_2\text{.}
\end{cases}
\]
\end{definition}

Any $\tau_0$-formula has an obvious interpretation in any
ordered $\Q$-vector space extending~$\Q$, and, in particular,
in~$\Q^\star$.

\begin{theorem}[\cite{VanDenDries}, Chapter~1, Remark~7.9]\label{equiv}
The first-order theory of ordered $\Q$-vector
spaces  in the signature~$\tau_0\cup\{+,-\}$ is complete.
\end{theorem}
Let $\mfa^\star$ be the $\tau$-structure obtained by
interpreting each relation symbol $R\in\tau$ by the
relation~$R^{\mfa^\star}$ defined on~$\Q^\star$ by the same (quantifier-free)
$\tau_0$-formula~$\phi_R$ that defines $R^\mfa$ over~$\Q$ (the choice of equivalent $\tau_0$-formulas is immaterial).
 Theorem~\ref{equiv} has the following immediate consequence.

\begin{corollary}\label{transfer}	
Let $\phi$ be a $\tau$-formula. Then $\phi$ is satisfiable in $\mfa$ if, and only if,~$\phi$ is satisfiable in $\mfa^\star$.
\end{corollary}

Our goal is to prove the following theorem.

\begin{theorem}\label{thm:sampling}
There is an efficient sampling algorithm for $\mfa^\star$.
\end{theorem}

Before giving the proof of Theorem \ref{thm:sampling} we present some preliminary lemmas in which we explicitly define the domain of the wanted sample. 
Let $\phi$ be an atomic $\tau_0$-formula. We write $\bar \phi$ for the formula $t_1 \leq t_2$ if $\phi$ is of the form $t_1 < t_2$, and for the formula $t_1 = t_2$ if $\phi$ is of the form $t_1 = t_2$. We call $\bar \phi$ the \emph{closure} of the formula $\phi.$
First, we investigate the positive solutions to the closures of finitely many atomic $\tau_0$-formulas. Then, in a second step that builds on the first one, we investigate the solutions to  finitely many  atomic $\tau_0$-formulas.

\begin{lemma}\label{closure}
	Let $\Phi$ be a finite set of atomic $\tau_0$-formulas having free variables
	in~$\{v_1,\dotsc, v_d\}$. Assume that $\bar \Phi:=\bigcup_{\phi\in \Phi} \bar\phi$
	has a simultaneous solution~$(x_1,\dotsc, x_d)\in\Q_{>0}$ in positive
	numbers. Then $\bar \Phi$ has a solution taking values in the
	set $C_{\Phi,d}\subset \Q$ defined as follows
	\[
	C_{\Phi,d} = 
	\left\{ |k|\prod_{i=1}^s |h_i| ^{e_i} \,\middle|\, k\in K(\Phi) ,\;
	e_1,\dotsc, e_s \in {\mathbb Z},\; \sum_{r=1}^s \lvert e_r \rvert < d
	\right\}
	\]
	where $h_1,\dotsc, h_s$ is an enumeration of the (finitely many) elements
	of~$H(\Phi)$.
\end{lemma}
\begin{proof}
	Let $\gamma\le\beta$ be maximal such that there are $\Psi_1,\Psi_2,\Psi_3$ with
	\begin{align*}
	\bar \Phi \mkern 7mu &= \{s_1 = s'_1, \dotsc, s_\alpha = s'_\alpha\} \cup
	\{ t_1 \le t'_1, \dotsc, t_\beta \le t'_\beta \}\\
	\Psi_1 &= \{s_1 = s'_1, \dotsc, s_\alpha = s'_\alpha\}\\
	\Psi_2 &= \{ t_1 = t'_1, \dotsc, t_\gamma=t'_\gamma \}\\
	\Psi_3 &= \{t_{\gamma+1} \le t'_{\gamma+1},\dotsc, t_\beta \le t'_\beta \}\text{, }
	\end{align*}
	where $s_i, s'_i, t_j, t'_j$ are $\tau_0$-terms for all $i$, $j$, and $\Psi_1 \cup \Psi_2 \cup \Psi_3$ is satisfiable in positive numbers.
	Clearly the space of positive solutions of~$\Psi_1\cup \Psi_2$ must be contained in
	that of~$\Psi_3$. In fact, by construction, they intersect: consider any straight
	line segment connecting a solution of $\Psi_1\cup \Psi_2\cup \Psi_3$ and a solution of
	$\Psi_1\cup \Psi_2$ not satisfying~$\Psi_3$, on this segment there must be a solution of
	$\Psi_1\cup \Psi_2\cup \Psi_3$ lying on the boundary of one of the inequalities of~$\Psi_3$,
	contradicting the maximality of~$\gamma$.
	By the last observation it suffices to prove that there is a solution
	of~$\Psi_1\cup \Psi_2$ taking values in~$C_{\Phi,d}$. Put an edge between two variables $x_i$ and~$x_j$ when
	they appear in the same formula of~$\Psi_1\cup \Psi_2$. For each connected component
	of the graph thus defined, either
	it contains at least one variable~$v_i$ such that there is
	a constraint of the form~$h\cdot v_i = k\cdot 1$, 
	or all constraints are of the form~$h\cdot v_i = h'\cdot v_j$.
	In the first case assign $v_i=\frac{k}{h}$, in the second assign one of the
	variables~$v_i$ arbitrarily to~$1$, then, in any case,
	since the diameter of the connected component
	is smaller than~$d$, all variables in this connected component are forced to take values in~$C_{\Phi,d}$ by
	simple propagation of~$v_i$.
\end{proof}

\begin{lemma}\label{pippo}
	Let $\Phi$ be a finite set of atomic $\tau_0$-formulas having free variables
	in~$\{v_1,\dotsc, v_d\}$. Assume that the formulas in~$\Phi$ are simultaneously
	satisfiable in $\Q$. Then they are simultaneously satisfiable in
	\[D_{\Phi,d}:=-C^{\star}_{\Phi,d} \cup \{0\} \cup C^{\star}_{\Phi,d}\subseteq\Q^\star\] where
	\[
	C^{\star}_{\Phi,d} = \{x + nx\boldsymbol\epsilon \;|\; x \in C_{\Phi,d},\; n\in\mathbb Z,\; -d\le
	n\le d\}\subseteq \Q^\star
	\]
	$C_{\Phi,d}$ is defined as in Lemma~\ref{closure}, and $-C^{\star}_{\Phi,d}$
	denotes the set~$\{-x\;|\;x\in C^{\star}_{\Phi,d}\}$.
\end{lemma}
\begin{proof}
	First we fix a solution $v_i=a_i$ for $i=1\dotsc d$ of~$\Phi$. In general,
	some of the values~$a_i$ will be positive, some~$0$, and some negative: we
	look for a new solution $z_1,\dotsc, z_d\in D_{\Phi,d}$ such that $z_i$ is
	positive, respectively~$0$ or negative, if and only if $a_i$ is.
	To this aim we rewrite the formulas in~$\Phi$ replacing each variable $v_i$
	with either~$y_i$, or~$0$ (formally $0\cdot 1$), or~$-y_i$ (formally
	$-1\cdot y_i$). We call~${\Phi^+}$ the new set of formulas, which, by
	construction, is satisfiable in positive numbers $y_i = b_i$. To establish
	the lemma, it suffices to find a solution of~${\Phi^+}$ taking values
	in~$C^{\star}_{\Phi,d}$.
	
	By Lemma~\ref{closure}, we have an assignment
	$y_i = c_i$ of values $c_1,\dotsc, c_d$ in $C_{{\Phi^+},d} \subseteq C_{\Phi,d}$ that satisfies
	simultaneously all formulas~$\bar\phi$ with~$\phi\in {\Phi^+}$.
	Let $ n_1,\dotsc, n_d \in \{n \in \mathbb Z  \mid -d \le n \le d\}$ be such that for all
	$i,j$
	\begin{align*}
	n_i < n_j \quad &\text{if and only if} \quad {\textstyle \frac{b_i}{c_i} <\frac{ b_j}{c_j}}\\
	0 < n_i \quad &\text{if and only if} \quad {\textstyle 1 <\frac{b_i}{c_i}}\\
	n_i < 0 \quad &\text{if and only if} \quad {\textstyle \frac{b_i}{c_i} < 1}.
	\end{align*}
	Such numbers exist: simply sort the set
	$\{1\}\cup\big\{\frac{b_i }{ c_i}\;|\;i=1,\dotsc,d\,\big\}$ and consider the positions in
	the sorted sequence counting from that of~$1$.
	We claim that the assignment
	$y_i = c_i + n_ic_i\boldsymbol\epsilon\in\Q^\star$ satisfies all formulas
	of~${\Phi^+}$.
	To check this, we consider the different cases for atomic formulas
	\begin{itemize}
		\item $k\cdot y_i < h \cdot y_j$: if $k c_i < h c_j$ this is
		obviously satisfied. Otherwise $k c_i = h c_j$, in this case $k$ and~$h$
		are positive and the
		constraint
		\[kc_i+kn_ic_i\boldsymbol\epsilon < hc_j+hn_jc_j\boldsymbol\epsilon\]
		is equivalent to~$n_i<n_j$. This, in turn, is equivalent by construction
		to $\frac{b_i}{c_i} < \frac{b_j}{c_j}$ which we get by observing that $b_ihc_j = b_ikc_i <
		b_jhc_i$.
		\item $k\cdot y_i = h \cdot y_j$: obviously $k b_i = h b_j$ and $k c_i = h
		c_j$, therefore $\frac{b_i}{c_i} = \frac{b_j}{c_j}$, and, as a consequence, also $n_i=n_j$
		from which the statement.
		\item $k\cdot 1 < h \cdot y_j$: similarly to the first case,
		if $k<h c_j$ this is immediate. Otherwise $k = hc_j$, so $k$ and~$h$ are
		positive, the constraint
		\[k\cdot 1 < hc_j+hn_jc_j\boldsymbol\epsilon\]
		is equivalent to
		$0 < n_j$, in other words $1<\frac{b_j}{c_j}$, which follows observing that
		$hc_j = k < hb_j$.
		\item $k\cdot y_i < h \cdot 1$: as the case above.
		\item $k\cdot 1 = h \cdot y_j$: obviously $k\cdot 1 = h b_j = h c_j$,
		therefore $\frac{b_j}{c_j} = 1$, so $n_j=0$ and the case follows.
		\item $k\cdot y_i = h \cdot 1$: as the case above.
	\end{itemize}
\end{proof}

\begin{proof}[of Theorem~\ref{thm:sampling}]
On input $d$, the sampling algorithm produces the finite
substructure~$\mfa^\star_{\atm{\tau},d}$ of~$\mfa^\star$ having
domain~$D_{\atm{\tau},d}$ where
$\atm{\tau}:=\bigcup_{R\in\tau}\atm{\phi_R}$, that is, the $\tau$-structure with
domain~$D_{\atm{\tau},d}$ in which each relation symbol~$R\in\tau$
denotes the restriction of~$R^{\mfa^{\star}}$ to~$D_{\atm{\tau},d}$. It is
immediate to observe that this structure can be computed in polynomial time  in~$d$.
Since $\mfa^\star_{\atm{\tau},d}$ is a substructure of~$\mfa^\star$, it is
clear that if an instance is satisfiable in~$\mfa^\star_{\atm{\tau},d}$,
then it is a fortiori satisfiable in~$\mfa^\star$.

The vice versa follows from Lemma~\ref{pippo}. In fact, consider a set
$\Psi$ of atomic $\tau$-formulas having free variables $x_1,\dotsc, x_d$. Assume
that $\Psi$ is satisfied in~$\mfa^\star$ by one assignment~$x_i=a_i$ for~$i\in \{1,\dotsc, d\}$. For
each~$R\in \Psi$ let $\Phi_R \subset \atm{\phi_R}$ be the set of atomic subformulas
of~$\phi_R$ which are satisfied by our assignment~$a_i$. Clearly the atomic
$\tau_0$-formulas $\Phi:=\bigcup_{R\in \Psi} \Phi_R$ are simultaneously satisfiable.
Remembering that the formulas~$\phi_R$ have no negations by construction,
it is obvious that any simultaneous solution of~$\Phi$ must also
satisfy~$\Psi$.
By Lemma~\ref{pippo}, $\Phi$ has a solution in the set~$D_{\Phi,d}$ defined
therein. We can observe
that $C_{\Phi,d}\subset C_{\atm{\tau},d}$,
hence $D_{\Phi,d}\subset D_{\atm{\tau},d}$ and the claim follows.
\end{proof}

From Theorem~\ref{thm:sampling} it easily follows that there is an efficient sampling algorithm for $\mfa$.

\begin{corollary}\label{cor:PLHsamplingcsp}
There exists an efficient sampling algorithm for $\mfa$. 
\end{corollary}

\begin{proof}
On input $d$, the sampling algorithm produces the $\tau_0$-reduct of the sample for $\mfa^\star$ (on input $d$) described in the proof of Theorem~\ref{thm:sampling}. By Corollary~\ref{transfer} and Theorem~\ref{thm:sampling}, this is an efficient sampling algorithm for $\mfa$.
\end{proof}

\begin{definition}
	A $k$-ary operation $g \colon \Q^k \to \Q$ is called a  \emph{polymorphism of $\mathfrak A$} if for all $R \in \tau$ we have that  $g(x^1,\ldots,  x^k) \in R^{\mathfrak A}$ for all $ x^1,\ldots,  x^k \in R ^{\mathfrak A}$, namely $R^{\mathfrak A}$ is {\it preserved} by~$g$ (where $g$ is applied component-wise).
\end{definition}
 

As an application of Theorem \ref{thm:sampling} 
we prove the following result:

\begin{theorem}\label{thm-cspmaxclosed}
	Let $\mathfrak A$ be a  PLH relational structure that is preserved by~$\max$.
	Then $\csp(\mathfrak A)$ is polynomial-time solvable.
\end{theorem}

This result is incomparable to known results
about  semilinear relations preserved by $\max$~\cite{BodirskyMaminoTropical}.
In particular, there, the weaker bound~$\text{NP}\cap\text{co-NP}$ has been shown for a larger class, and polynomial tractability only for a smaller class (which does not contain many  PLH relations preserved by $\max$, for instance $x \geq \max(y,z)$). To prove Theorem \ref{thm-cspmaxclosed} we use the notion of totally symmetric polymorphism and a result from \cite{BodMacpheTha}.
\begin{definition}\label{def:totsym}A $k$-ary operation $g$ is \emph{totally symmetric} if \[g(x^1,\ldots,x^k)=g(y^1,\ldots,y^k) \text{ whenever } \{x^1,\ldots,x^k\}=\{y^1,\ldots,y^k\}.\]\end{definition}
\begin{example}\label{example:karymax}
For every $k\geq 1$, the $k$-ary totally symmetric operation $\max^{(k)}$ defined by 
	$(x^1,\ldots,x^k)\mapsto
		\max(x^1,\ldots,x^k)$
		is totally symmetric.\demo
\end{example}
\begin{restatable}[\cite{BodMacpheTha}, Theorem 2.5]{theorem}{macphe}
	\label{macphe}
	Let $\mathfrak A$ be a structure over a finite relational
	signature with totally symmetric polymorphisms of all arities. If there
	exists an efficient sampling algorithm for $\mathfrak A$ then
	$\csp({\mathfrak A})$ is in P.
\end{restatable}




 \begin{proof}[of Theorem~ \ref{thm-cspmaxclosed}]
			Since $\mathfrak A$ is preserved by $\max$, 	for every $k \geq 1$, the $k$-ary totally symmetric operation $\max^{(k)}$ (see Example \ref{example:karymax})
		 preserves $\mathfrak A$. To see this, observe that for every $x^1,\ldots,x^k \in \Q$, it holds that  \[\textstyle \max^{(k)}(x^1,\ldots,x^k)=\max(x^1, \max(x^2, \ldots, \max(x^{k-1},x^k)\ldots)).\]
		Therefore,  as  by  Corollary~\ref{cor:PLHsamplingcsp} there exists an efficient sampling algorithm for $\mfa$, from Theorem~\ref{macphe} it follows that $\csp(\mfa)$ can be solved in polynomial time.
\end{proof}

\section{Efficient Sampling for PLH Valued   Structures}
\label{sect:tract}
In this section we introduce the notion of a sampling algorithm for valued structures and we exhibit an efficient sampling algorithm for PLH valued structures. Analogously to the case of relational structures, before  defining the notion of a sampling algorithm for valued structures, we define a generalisation of the universal algebraic concept of a homomorphism  to the valued constraint satisfaction framework. 

\begin{definition}
	Let $\Delta$ and $\g$ be valued structures with the same signature $\tau$ with domain $D$ and $C$ respectively. Let $C^D$ denote the set of all functions $g \colon D \to C$. A \emph{fractional homomorphism} \cite{ThapperZivny2012}  from $\Delta$ to $\g$ is a function $\omega \colon C^D \to Q_{\geq0}$ with finite support, $\supp(\omega):=\{g \in C^D \mid \omega(g)>0 \}$, such that  $\sum_{g \in C^D} \omega (g)=1$, and such that for every function symbol $f \in \tau$ and tuple $a \in D^{\ar(f)}$, it holds that \[\sum_{g \in C^D}\omega(g)f^{\g}(g(a))\leq f^{\Delta}(a),\]
	where the functions $g$ are applied componentwise. We say that $\Delta$ is \emph{fractionally homomorphic} to $\g$ and write $\Delta \to_f \g$ to indicate the existence of a fractional homomorphism from $\Delta$ to $\g$.
\end{definition}

Let $\g$ be a valued structure with domain $C$ and signature $\tau$. A \emph{valued substructure} $\Delta$ of $\g$ is a valued structure with domain $D \subseteq C$, with signature $\tau$, and such that the interpretation of a function symbol $f$ from $\tau$ is the restriction of $f^{\g}$ to the domain $D$. 
Trivially, if $\g$ is a valued structure and  $\Delta$ is a valued substructure  of $\g$, then $\Delta$ is fractionally homomorphic to $\g$.

\begin{definition}\label{def:valued-sampling}
	Let $\g$ be a valued structure with domain $C$ and finite signature $\tau$. A \emph{sampling algorithm for $\g$} takes as input a positive integer $d$  and computes a finite-domain valued $\tau$-structure  $\Delta$ fractionally homomorphic to $\g$ 
	such that, for every finite sum $\phi$ of $\tau$-terms having at most $d$ distinct variables, $V=\{x_1,\ldots,x_d\}$, and every $u \in \Q$,  there exists a solution $h\colon V \to C$ with $\phi^{\g} (h(x_1),\ldots, h(x_d))\leq u$ if, and only if, there exists a solution $h'\colon V \to D$ with $\phi^{\Delta}(h'(x_1),\ldots, h'(x_d))\leq u$.
	A sampling algorithm is called \emph{efficient} if its running time is bounded by a polynomial in $d$.
\end{definition}  

The valued finite structure computed by a sampling algorithm is called a \emph{sample}.

\begin{remark}
	Observe that the output $\Delta_d$ of a sampling algorithm for a given valued structure $\g$ with finite signature $\tau$ does not depend on the rational threshold $u$. Therefore, the given sampling algorithm has the property that given a finite sum $\phi$ of function symbols from $\tau$ with variables $V:=\{x_1,\ldots,x_n\}$, it holds that \[\inf_{h \colon V \to \dom(\g)}\phi^{\g}(h(x_1),\ldots,h(x_n))=\inf_{h' \colon V \to \dom(\Delta_d)}\phi^{\Delta_d}(h'(x_1),\ldots,h'(x_n)).\]
\end{remark} 
\subsection{The Ring of Formal Laurent Power Series}\label{subsect:Laurentpwr}

In the present section we extend the method developed in Section~\ref{sect:csp-tract} to the
treatment of VCSPs. To better highlight the analogy with
Section~\ref{sect:csp-tract}, so that the reader already familiar with it
may quickly get an intuition of the arguments here, we use identical
notation to represent corresponding objects. This choice has the drawback that some symbols, notably~$\Q^\star$, need to be re-defined (the new $\Q^\star$ will contain the old one). In this section, we will sometimes skip details that can be borrowed unchanged from Section~\ref{sect:csp-tract}.


\begin{definition}
Let $\Q^\star$ denote the ring~$\Q((\eps))$ of formal Laurent
power series in the indeterminate~$\eps$, that is, $\Q^\star$ is the set of
formal expressions
\[\sum_{i=-\infty}^{+\infty} a_i\eps^i\]
where $a_i\neq0$ for only finitely many {\it negative} values of~$i$.
Clearly, $\Q$ is embedded in~$\Q^\star$ (the embedding is defined similarly to its corresponding in Section \ref{sect:csp-tract}).
The ring operations on $\Q^\star$ are defined as usual
\begin{align*}
\sum_{i=-\infty}^{+\infty} a_i\eps^i + \sum_{i=-\infty}^{+\infty} b_i\eps^i &=
\sum_{i=-\infty}^{+\infty} (a_i+b_i) \eps^i \\
\sum_{i=-\infty}^{+\infty} a_i\eps^i \cdot \sum_{i=-\infty}^{+\infty} b_i\eps^i &=
\sum_{i=-\infty}^{+\infty} \left(\sum_{j=-\infty}^{+\infty} a_jb_{i-j}\right) \eps^i
\end{align*}
where the sum in the product definition is always finite by the hypothesis
	on $a_i,b_i$ with negative index~$i$. The order is the lexicographical
	order induced by~$0<\eps\ll1$, i.e.,
	\[
	\sum_{i=-\infty}^{+\infty} a_i\eps^i < \sum_{i=-\infty}^{+\infty} b_i\eps^i
	\quad\text{iff}\quad \exists i\; a_i < b_i \wedge \forall j<i\; a_j=b_j
	\text{.}
	\]
	It is well known that $\Q^\star$ is an ordered field, that is, all non-zero
	elements have a multiplicative inverse and the order is
	compatible with the field operations. We define the following subsets
of~$\Q^\star$ for $m\le n$
\[\Q^\star_{m,n} := \left\{ \sum_{i=m}^n \boldsymbol\epsilon^i a_i
\;\middle|\; a_i\in\Q \right\} \subset \Q^\star
.\]
\end{definition}

\begin{definition}\label{def:Lstar}
We define a new structure $\mathfrak L^\star$, which is  an extension of
an expansion (or an expansion of an extension)  of~$\mathfrak L$, having $\Q^\star$ as domain and
$\tau_1 := \tau_0\cup\{k\}_{k\in\Q^\star_{-1,1}}$
as signature,
where the interpretation of
symbols in~$\tau_0$ is formally the same as for~$\mathfrak L$ and
the symbols $k\in\Q^\star_{-1,1}$ denote constants (zero-ary functions).
\end{definition}
Notice that, for technical reasons, we allow only constants
from $\Q^\star_{-1,1}$. In the remainder of this section,
$\tau_1$-formulas will be interpreted in the structure~$\mathfrak
L^\star$. We make for $\tau_1$-formulas the same assumptions as in
Section~\ref{sect:csp-tract} (that atomic subformulas are of the form  either $c_1\cdot 1 {<\atop =} x_i$, or $x_i {<\atop =} c_2\cdot 1$, or $c_1\cdot x_i {<\atop =} c_2\cdot x_j$ with constants~$c_1$ and~$c_2$  not both negative and where function symbols $c_i\cdot$ are never composed). Also $H(\Phi)$ and~$K(\Phi)$, where $\Phi$ is a set of atomic
$\tau_1$-formulas, are defined similarly to Section~\ref{sect:csp-tract}, and the only difference is that now $\Phi$ is a set of atomic
$\tau_1$-formulas rather than $\tau_0$-formulas.
Observe that the reduct of $\mathfrak L^\star$ obtained by
restricting the signature to~$\tau_0$ is {\it elementarily equivalent}
to~$\mathfrak L$, namely it satisfies the same first-order sentences.

Similarly as in Section \ref{sect:csp-tract}, for every valued PLH structure $\g$ and for every positive integer number $d$, we explicitly give a $\tau_1$-structure  with finite domain $D^\star$ such that  every instance of $\vcsp(\g)$ with at most $d$ distinct free variables has a solution with values in $\Q$ if, and only if, it has a solution with values in $D^\star\subseteq \Q^\star$ (Lemma~\ref{vcsp-sampling}). We will need two preliminary results: Lemma~\ref{closure2} and Lemma~\ref{pippo2}, which are analogues of Lemma~\ref{closure} and Lemma~\ref{pippo} from Section~\ref{sect:csp-tract}. More specifically, in Lemma~\ref{closure2} we consider the positive solutions to the closures of finitely  $\tau_1$-formulas, and  in Lemma~\ref{pippo2} we consider the positive solutions to finitely many $\tau_1$-formulas.

\begin{lemma}\label{closure2}
Let $\Phi$ be a finite set of atomic $\tau_1$-formulas having free variables
in~$\{v_1,\dotsc, v_d\}$
and let~$\bar\Phi$ be the set~$\bigcup_{\phi\in\Phi}\bar\phi$.
Suppose that there is $0<r\in\Q^\star$ such that
all satisfying assignments of~$\bar\Phi$ with values $(x_1,\ldots,x_d)$ in~$\Q^\star$
also satisfy $0< x_i\le r$ for all~$i$. Let $u$, $\alpha_1,\dotsc,\alpha_d$
be elements of~$\Q^\star$. Assume that the formulas in~$\Phi$ are simultaneously
satisfiable by a point~$(x_1,\dotsc, x_d)\in {(\Q^\star)}^d$ such that
$\sum_i \alpha_ix_i < u$.
Let us define the set
\[
C_{\Phi,d} = 
\left\{|k|\prod_{i=1}^s |h_i| ^{e_i} \,\middle|\, k\in K(\Phi) ,\;
e_1,\dotsc, e_s \in {\mathbb Z},\; \sum_{r=1}^s \lvert e_r \rvert < d
\right\} \subseteq \Q^\star_{-1,1}
\]
where $h_1,\dotsc, h_s$ is an enumeration of the (finitely many) elements
of~$H(\Phi)$.
Then there is
a point in $(x_1',\dotsc, x_d')\in (C_{\Phi,d})^d\subseteq{(\Q^\star)}^d$
with $\sum_i \alpha_ix_i' < u$ that satisfies simultaneously
all~$\bar\phi$, for~$\phi\in\Phi$.
\end{lemma}
\begin{proof}
As in the proof of Lemma~\ref{closure} (to which we direct the reader for
many details) we take a maximal~$\gamma\le\beta$ such that there are
$\Psi_1,\Psi_2,\Psi_3$ with
\begin{align*}
\bar \Phi \mkern 7mu &= \{s_1 = s'_1, \dotsc, s_\alpha = s'_\alpha\} \cup
\{ t_1 \le t'_1, \dotsc, t_\beta \le t'_\beta \}\\
\Psi_1 &= \{s_1 = s'_1, \dotsc, s_\alpha = s'_\alpha\}\\
\Psi_2 &= \{ t_1 = t'_1, \dotsc, t_\gamma=t'_\gamma \}\\
\Psi_3 &= \{t_{\gamma+1} \le t'_{\gamma+1},\dotsc, t_\beta \le t'_\beta \}
\end{align*}
and $\Psi_1 \cup \Psi_2 \cup \Psi_3$ is satisfiable by an assignment $(x_1,\ldots,x_d)$
with~$\sum_i \alpha_ix_i < u$. As in the proof of Lemma~\ref{closure}
the set of solutions of~$\Psi_1\cup\Psi_2$ satisfying~$\sum_i \alpha_ix_i
< u$ is contained in the solutions of~$\Psi_3$. So, here too, it
suffices to show that there is a solution of~$\Psi_1\cup\Psi_2$
with~$\sum_i \alpha_iv_i < u$ taking values in~$C_{\Phi,d}$. The proof
of Lemma~\ref{closure} shows that there is a solution 
of~$\Psi_1\cup\Psi_2$ taking values $(x_1,\ldots,x_d)$ in~$C_{\Phi,d}$ without necessarily
meeting the requirement that~$\sum_i \alpha_ix_i < u$. We will prove that, in
fact, any such solution meets the additional constraint.

Let $x_i=a_i,b_i$ be two distinct satisfying assignments
for~$\Psi_1\cup\Psi_2$ such that $\sum_i \alpha_ia_i < u$ and
$\sum_i \alpha_ib_i \ge u$. We know that the first exists, and we assume
the second towards a contradiction. The two assignments must differ, so,
without loss of generality~$a_1\neq b_1$.
For $t\in\Q^\star$, with $t\ge0$, define the assignment
$x_i(t)=(1+t)a_i-tb_i$.
Since all constraints in~$\Psi_1\cup\Psi_2$ are
equalities, it is clear that the new assignment~$x_i(t)$
satisfies~$\Psi_1\cup\Psi_2$ for all~$t\in\Q^\star$.
Moreover, if $t\ge0$
\[
\sum_i \alpha_ix_i(t) \le \sum_i \alpha_ia_i - t\left(\sum_i \alpha_ib_i -
\sum_i \alpha_ia_i\right) < u
\]
Let $t=\frac{2r}{|b_1-a_1|}$. Then
\[
x_1(t) = a_1 + \frac{2r}{|b-a|}(a-b)
\]
is either not smaller than $2r$ or smaller than $0$ depending on the sign of~$(a-b)$. In either
case we have a solution~$x_i=x_i(t)$ of~$\Psi_1\cup\Psi_2$ satisfying~$\sum_i
\alpha_ix_i(t) < u$, which must therefore be a solution of~$\Phi$
that does not satisfy $0< x_i\le r$.
\end{proof}

\begin{lemma}\label{pippo2}
Let $\Phi$ be a finite set of atomic $\tau_1$-formulas having free variables
in~$\{v_1,\dotsc, v_d\}$. Suppose that there are $0<l<r\in\Q^\star$ such that
all satisfying assignments $(x_1,\ldots,x_d)$ of~$\Phi$ in the domain~$\Q^\star$
also satisfy $l< x_i< r$ for all~$i$. Let $\alpha_1,\dotsc,\alpha_d$
be rational numbers and~$u\in\Q^\star_{-1,1}$. Assume that the formulas in~$\Phi$ are simultaneously
satisfiable by a point~$(x_1,\dotsc, x_d)\in (\Q^\star)^d$ such that
$\sum_i \alpha_ix_i \le u$. Then the same formulas
are simultaneously satisfiable by a point $(x_1',\dotsc, x_d')\in
(C^\star_{\Phi,d})^d\subseteq(\Q^\star)^d$ such that $\sum_i \alpha_ix_i' \le u$
where
\[
C^\star_{\Phi,d} = \{x + nx\boldsymbol\epsilon^3 \;|\; x \in C_{\Phi,d},\;
n\in\mathbb Z,\; -d\le
n\le d\} \subseteq \Q^\star_{-1,4} .
\]
\end{lemma}
\begin{proof}
We consider two cases: either
all satisfying assignments $(x_1,\ldots,x_d)$ satisfy the inequality $\sum_i \alpha_ix_i \ge u$,
or 
there exists a satisfying assignment $(x_1,\dotsc,
x_d)$ for~$\Phi$ such that $\sum_i \alpha_ix_i < u$.

In the first case, we
claim that all satisfying assignments, in fact, satisfy $\sum_i
\alpha_ix_i = u$. In fact, assume that $x_i=a_i,b_i$ are two satisfying
assignments such that $\sum_i \alpha_ia_i = u$ and~$v := \sum_i
\alpha_ib_i > u$. As in the proof of Lemma~\ref{closure2}, consider
assignments of the form~$x_i(t)=(1+t)a_i - t b_i$
for $t\in\Q^\star$. Clearly $\sum_i \alpha_ix_i(t) = u - t (v-u) < u$ for
all~$t>0$. As in Lemma~\ref{closure2}, the new assignment must
satisfy all equality constraints in~$\Phi$. Each inequality
constraint implies a strict inequality on~$t$ (remember that $\Phi$ only
has strict inequalities). Since all of these must be satisfied by~$t=0$,
there is an open interval of acceptable values of~$t$ around~$0$, and,
in particular, an acceptable~$t>0$. Our claim is thus established. 
Therefore, in this case,
it suffices to find any satisfying assignment for~$\Phi$
taking values in~$C^\star_{\Phi,d}$. The assignment is now constructed as in the
proof of Lemma~\ref{pippo}, replacing the formal symbol~$\boldsymbol\epsilon$ in that
proof by~$\boldsymbol\epsilon^3$. Namely take a satisfying
assignment~$x_i = b_i$ for~$\Phi$, and, by Lemma~\ref{closure2}, one
satisfying assignment~$x_i=c_i$ for~$\bar\Phi$ taking values
in~$C_{\Phi,d}$.
Observe that the hypothesis that all solutions of~$\Phi$ satisfy~$l<x_i$
for all~$i$ is used here to ensure that all solutions of~$\bar\Phi$ assign
positive values to the variables, which is required by
Lemma~\ref{closure2}.
Let $-d\le n_1,\dotsc, n_d\le d$ be integers such that for
all
$i,j$
\begin{align*}
n_i < n_j \quad &\text{if and only if} \quad {\textstyle \frac{b_i}{c_i}
<\frac{ b_j}{c_j}}\\
0 < n_i \quad &\text{if and only if} \quad {\textstyle 1
<\frac{b_i}{c_i}}\\
n_i < 0 \quad &\text{if and only if} \quad {\textstyle \frac{b_i}{c_i} <
1}
\end{align*}
The assignment
$y_i = c_i + n_ic_i\boldsymbol\epsilon^3$ can be seen to satisfy
all formulas of~$\Phi$ by the same check as in the proof of
Lemma~\ref{pippo}. Observe that we have to replace $\eps$
in Lemma~\ref{pippo} by~$\eps^3$ here, so
that~$\Q^\star_{-1,1} \cap \,\eps^3 \Q^\star_{-1,1} = \emptyset$.

For the second case, fix a satisfying assignment $x_i=b_i$.
By Lemma~\ref{closure2} there is an
assignment~$x_i=c_i\in C_{\Phi,d}$ such that $\sum_i \alpha_ic_i < u$
and this assignment satisfies~$\bar\phi$ for all~$\phi\in\Phi$. From these
two assignments construct the numbers~$n_i$ and then the assignment~$y_i =
c_i + n_ic_i\boldsymbol\epsilon^3$ as before. For the same reason
it is clear that the new assignment satisfies~$\Phi$. To conclude that
$\sum_i \alpha_iy_i < u$ we write
\[
\sum_i \alpha_iy_i = \sum_i \alpha_ic_i + \boldsymbol\epsilon^3
\sum_i \alpha_in_ic_i < u
\]
because the first summand is in~$\Q^\star_{-1,1}$ and smaller than~$u$, therefore the
second summand is neglected in the lexicographical order.
\end{proof}

\begin{lemma}\label{vcsp-sampling}
Let $\Phi$ be a finite set of atomic $\tau_0$-formulas having free variables
in~$\{v_1,\dotsc, v_d\}$.
Let $u$, $\alpha_1,\dotsc,\alpha_d$
be rational numbers. Then the following are equivalent:
\begin{enumerate}
\item The formulas in~$\Phi$ are simultaneously
satisfiable in $\Q$, by a point $(x_1,\dotsc, x_d)\in\Q^d$ such that
$\sum_i \alpha_ix_i \le u$.
\item The formulas in~$\Phi$ are simultaneously
satisfiable in $D_{\Phi,d}\subseteq\Q^\star$,
by a point~$(x_1',\dotsc, x_d')\in D_{\Phi,d}^d$
such that $\sum_i \alpha_ix_i' \le u$,
where the set $D_{\Phi,d}$ is defined as follows
\begin{align*}
D_{\Phi,d} &:= -C^{\star}_{\Phi',d} \cup \{0\} \cup C^{\star}_{\Phi',d} 
\subseteq \Q^\star_{-1,4} \\
\Phi' &:= \Phi \cup \{x > \eps,\;x < -\eps,\;x > -\eps^{-1},\;x <
\eps^{-1}\}.
\end{align*}
\end{enumerate}
\end{lemma}

\begin{remark}\label{rem:samplesize}
Observe that the set $D_{\Phi,d}$ depends neither on the $\alpha_i$'s nor on $u$. In fact, it only depends on the set of formulas $\Phi$ and on the number $d$ of free variables.
\end{remark}

\begin{proof}
The implication $\textit{2}\rightarrow\textit{1}$\/ is immediate observing
that the conditions~$\Phi$ and~$\sum_i \alpha_ix_i \le u$ are
first-order definable in~$\mathfrak S$. In fact, any assignments with
values in~$D_{\Phi,d}$ satisfying the conditions is, in particular, an
assignment in~$\Q^\star$, and, by the completeness of the first order theory
of ordered $\Q$-vector spaces, we have an assignment taking values
in~$\Q$.

For the vice versa, fix any assignment $x_i=a_i$ with $a_i\in\Q$
for~$i\in\{1,\dotsc, d\}$. We pre-process the formulas in $\Phi$ producing a
new set of atomic formulas~$\Phi'$ as follows. We replace every variable
$x_i$ such that $a_i=0$ with the constant~$0=0\cdot1$. Then we replace
each of the remaining variables $x_i$ with either~$y_i$ or~$-y_i$
according to the sign of~$a_i$. Finally, we add the constraints $\eps<y_i$
and $y_i<\eps^{-1}$ for each of these variables. Similarly we produce
new coefficients $\alpha'_i = \text{sign}(a_i)\alpha_i$.
It is clear that
the new
set of formulas~$\Phi'$ has a satisfying assignment in positive rational
numbers with~$\sum_i \alpha'_iy_i \le u$.
Observing that
a positive rational~$x$ always satisfies~$\eps<x<\eps^{-1}$, we see that $\Phi'$ satisfies the
hypothesis of Lemma~\ref{pippo2} with $l=\eps$ and~$r=\eps^{-1}$. Hence the
statement.
\end{proof}

\begin{remark}
		Lemma  \ref{vcsp-sampling} provides a polynomial-time many-one reduction of the $\vcsp$ for a PLH valued structure $\g$ with finite signature  to the $\vcsp$ for a  valued structure  $\Delta^\star$ with finite signature having as domain a finite subset of $\Q^\star$. We want to point out that, however, Lemma  \ref{vcsp-sampling} does not give rise to a sampling algorithm: firstly, the computed finite domain is a subset of $\Q^\star$ rather than $\Q$, secondly, the signature of $\Delta^\star$ is strictly larger than the signature of $\g$. In the next section we show how to obtain an efficient sampling algorithm for PLH valued structures using Lemma \ref{vcsp-sampling}. 
\end{remark}

\subsection{Reduction to a VCSP over a Finite  Rational Domain} \label{subsec:rationalsample}
In this section we use the results achieved in Section \ref{subsect:Laurentpwr} to provide an efficient  sampling algorithm for PLH valued structures.

Let $\Phi$ be a finite set of $\tau_0$-formulas having at most $d$ distinct free variables. By Lemma \ref{vcsp-sampling}, for every $\alpha_1,\ldots,\alpha_d,u\in \Q$ the formulas in $\Phi$ are simultaneously satisfiable in $\Q$ by a point $(x_1,\ldots,x_d)$ such that $\sum_{i=1}^{d}\alpha_ix_i\leq u$ if, and only if, they are simultaneously satisfiable in 
 $D_{\Phi,d}\subseteq \Q_{-1,4}^\star$ by a point $(x'_1,\ldots,x'_d)$ such that $\sum_{i=1}^{d}\alpha_ix'_i\leq u$. The elements in $D_{\Phi,d}\subseteq \Q_{-1,4}^\star$ are of the form \[\sum_{i=-1}^{4}x_i\eps^i \; \text{where } x_i \in -C_{\Phi,d}\cup\{0\} \cup C_{\Phi,d} \]
where $C_{\Phi,d}$ is as in Lemma \ref{closure}.  In $\mathfrak L^\star$ it holds that\begin{align*}
-\eps^{-1}&<x<-\eps \; & \text{for every } &x \in  -C_{\Phi,d},\\
\eps&<x<\eps^{-1} \; & \text{for every } &x \in  C_{\Phi,d}.
\end{align*}
(See Lemmas \ref{closure2}, \ref{pippo2}, \ref{vcsp-sampling}, and Lemmas \ref{closure}, \ref{pippo}.)

\begin{lemma} \label{lemmaepsilon}
	Let $\Phi$ be a finite set of $\tau_0$-formulas having at most $d$ distinct free variables. Let $D_{\Phi,d}$ be defined as in Lemma \ref{vcsp-sampling} (2.) and let $\mathfrak D^\star$ be the finite substructure of $\mathfrak L^\star$ with domain $D_{\Phi,d}\subseteq \Q_{-1,4}^\star$. 
	Then there exists a positive rational $\varepsilon$ with $-\varepsilon^{-1}<x<-\varepsilon$ for every $x \in -C_{\Phi,d}$  and $\varepsilon<x<\varepsilon^{-1}$ for every $x \in C_{\Phi,d}$   such that the map $\eta\colon D_{\Phi,d}\to \Q$ defined by
\[	\sum_{j=1}^{4}x_j\epsilon^j \mapsto \sum_{j=1}^{4}x_j \varepsilon^j\]
	is  a homomorphism from the $\tau_0$-reduct of $\mathfrak D^\star$ to $\mathfrak L$. Moreover, $\varepsilon$ and $\eta$ are computable in time polynomial in $d$.
\end{lemma}

\begin{proof}
	Let $C_{\Phi,d}\subseteq \Q$ be as in  Lemma \ref{closure} and let $-C_{\Phi,d}:=\{-x \mid x \in C_{\Phi,d}\}$.
We define
 $C:=\{x-y \mid x,y \in C_{\Phi,d} \; \text{and } x-y >0\}$. Observe that $C \supseteq C_{\Phi,d}$. Define \[\varepsilon :=\frac{1}{6}\frac{\min C}{\max C}.\]
	The number $\varepsilon$ is positive and rational and it can be computed in time polynomial in $d$. Furthermore,\begin{align*}
	&0< \varepsilon<1,\\
	&\varepsilon<\min C<\min C_{\Phi,d}\;  &\text{implying } \max( -C_{\Phi,d})<{-\varepsilon} &\text{, and}  \\
	&\max C=\max C_{\Phi,d}<\varepsilon^{-1}\;  &\text{implying } {-\varepsilon^{-1}}<\min( -C_{\Phi,d}).&
	\end{align*} 
	It follows that $-\varepsilon^{-1}<x<-\varepsilon$ for every $x \in -C_{\Phi,d}$,  and $\varepsilon<x<\varepsilon^{-1}$ for every $x \in C_{\Phi,d}$.
	
It is easy to see that $\eta$ preserves the scalar multiplication by rational elements and the identity in $\Q$.  
	We prove now that $\eta$ is  order preserving (and therefore injective). Let us consider \[x:=\sum_{j=-1}^4x_j\eps^j \text{, and } y:=\sum_{j=-1}^4y_j\eps^j \in D_{\Phi,d}\] such that $x$ is smaller than $y$ in the lexicographic order induced by  $\Q^\star_{-1,4}$. This means that there exists an index $i \in \{-1,\ldots,4\}$ such that $x_j=y_j$ for $-1\leq j<i$ and $x_i<y_i$. Since $\varepsilon >0$, it holds that $x_i\varepsilon^i<y_i\varepsilon^i$ and, consequently, \[\sum_{j=-1}^{i}x_j \varepsilon^j<\sum_{j=1}^{i}y_j \varepsilon^j.\]  Moreover, if $i\neq 4$ then for all $j \in \{i+1,\ldots,4\}$  \[(x_j-y_j)\varepsilon^{j-i}\leq (x_j-y_j)\varepsilon\leq (\max C)\varepsilon  \leq \frac{\min C}{6}  \leq \frac{y_i-x_i}{6}  \] because $\varepsilon<1 $, and $x_j-y_j\leq \max C$ (indeed, if $x_j-y_j>0$ then $x_j-y_j \in C$, otherwise $x_j-y_j$ is smaller than any element in $C$). Therefore we get
	$(x_j-y_j)\varepsilon^j \leq \frac{y_i-x_i}{6}\varepsilon^i$, and \[\sum_{j=i+1}^{4}(x_j-y_j)\varepsilon^j\leq\sum_{j=i+1}^{4}\frac{y_i-x_i}{6}\varepsilon^i\leq  \frac{5}{6}(y_i-x_i)\varepsilon^i<(y_i-x_i)\varepsilon^i.\] It follows that \[\sum_{j=i}^{4}x_j\varepsilon^j<\sum_{j=i}^{4}y_j\varepsilon^j,\]
	and, because $x_j=y_j$ for $-1\leq j \leq i-1$,
	\[\sum_{j=-1}^{4}x_j\varepsilon^j<\sum_{j=-1}^{4}y_j\varepsilon^j,\]
	i.e., $\eta$ preserves the order. 
\end{proof}

Let $\Phi$ be a finite set of $\tau_0$-formulas with at most $d$ variables, and let $\eta$ be the map introduced in Lemma \ref{lemmaepsilon}, we set $E_{\Phi,d}:=\eta(D_{\Phi,d})\subseteq \Q$.  Lemma \ref{lemmaepsilon} implies the following corollary.
\begin{corollary}\label{cor:PLHsampling}
	Let $\g$ be a PLH valued structure with finite signature $\tau$. Then there exists an efficient sampling algorithm for $\g$. 
\end{corollary}

\begin{proof}
	For every cost function $f$ in $\tau$ let us consider the quantifier-free $\tau_0$-formula $\phi_{f}$ defining  $R_f^{\g}=\dom(f)$ over $\Q$. Let $\atm{\phi_{f}}$  denote the set of atomic subformulas of $\phi_{f}$ and let  $\atm{\tau}:=\bigcup_{f\in\tau}\atm{\phi_{f}}$. 
	
	On input $d$, the algorithm produces the 
		valued finite substructure $\Delta$ of $\g$ having domain $\eta(D_{\atm{\tau},d})$, where eta is defined as in Lemma \ref{lemmaepsilon}.  It is immediate to see that the valued structure $\Delta$ has size polynomial in $d$ and can be computed in time polynomial in $d$, because  $D_{\atm{\tau},d}$ and $\eta(D_{\atm{\tau},d})$  can be computed in polynomial time in $d$ (see Remark \ref{rem:samplesize}, and Lemma \ref{lemmaepsilon}).
	Let $\phi$ be a finite sum of function symbols from $\tau$ with at most $d$ variables from $V:=\{v_1,\ldots,v_d\}$, and let $u$ be a rational number. By Lemma \ref{vcsp-sampling} there exists an assignment $h \colon V\to \Q$ such that $\phi^{\g}(h)\leq u$ if, and only if, there exists an assignment $h' \colon V\to D_{\atm{\tau},d}$ such that $\phi^{\g^\star}(h)\leq u$. Furthermore,  by Lemma \ref{lemmaepsilon} there exists a positive rational number $\varepsilon$ such that the map $\eta:D_{\atm{\tau},d}\to \Q$ is a homomorphism of $\tau_0$-structures from the $\tau_0$-reduct of $\mathfrak D^\star$ to $\mathfrak L$. Since $\eta$ is injective,  the assignment  $\eta \circ h'\colon V \to  \eta(D_{\atm{\tau},d})$ has cost $\phi^{\Delta}(\eta \circ h')\leq u$.
\end{proof}

\section{Fully Symmetric Fractional Polymorphisms}\label{sect:totsym}
In this section we give sufficient algebraic conditions under which valued structures with an infinite
domain that admit an efficient sampling algorithm can be solved in polynomial time.
\subsection{Universal Algebraic Tools}
In this section we survey the concepts from universal algebra that we need and use in the remainder of the article.

\begin{definition}
Let $\Delta$ and $\g$ be valued structures with the same signature $\tau$ with domain $D$ and $C$ respectively. Let $C^D$ denote the set of all functions $g \colon D \to C$. A \emph{fractional homomorphism} \cite{ThapperZivny2012}  from $\Delta$ to $\g$ is a function $\omega \colon C^D \to Q_{\geq0}$ with finite support, $\supp(\omega):=\{g \in C^D \mid \omega(g)>0 \}$, such that  $\sum_{g \in C^D} \omega (g)=1$, and such that for every function symbol $f \in \tau$ and tuple $a \in D^{\ar(f)}$, it holds that \[\sum_{g \in C^D}\omega(g)f^{\g}(g(a))\leq f^{\Delta}(a),\]
where the functions $g$ are applied componentwise. We write $\Delta \to_f \g$ to indicate the existence of a fractional homomorphism from $\Delta$ to $\g$.
\end{definition}

The following proposition is adapted from \cite{ThapperZivny2012}, Proposition 2.1 where it is stated for valued structures with finite domains. In fact, we need not explicitly prove it, as the proof of the proposition in the finite-domain case can be easily modified to work in the infinite-domain case.

\begin{proposition} \label{prop:frachom}
	Let $\Delta$ and $\g$ be valued structures with the same signature $\tau$ and with domain $D$ and $C$, respectively. Assume that $\Delta \to_f \g$. Let $I$ be an instance of $\vcsp(\g)$ having variables $V_I=\{x_1,\ldots,x_n\}$, objective function $\phi_I(x_1,\ldots,x_n)$, and threshold $u_I \in \QQ$.  Suppose that there exists an assignment $h \colon V_I \to D$ such that $\phi_I^{\Delta}(h(x_1, \ldots, h(x_n)) \leq u_I$. Then there exists an assignment $h'\colon V_I \to C$ such that $\phi_I^{\g}(h'(x_1), \ldots, h'(x_n)) \leq u_I$. In particular, it holds that \[\inf_{c \in C^n} \phi_I^{\g}(c) \leq \inf_{d \in D^n} \phi_I^{\Delta}(d).\] 
\end{proposition}

Let $\g$ be a valued structure with countable (possibly infinite) domain $C$ and signature $\tau$. An $m$-ary \emph{operation} on $C$ is a function $g \colon C^m \to C$. Let $\mathcal O_C^{(m)}$ denote the set of all $m$-ary operations on $C$. \begin{definition}
An $m$-ary \emph{fractional operation} is a function $\omega\colon \mathcal O_C^{(m)} \to \Q_{\geq0}$. The set $\{g \in \mathcal O_C^{(m)} \mid \omega(g)>0\}$ is called the \emph{support} of $\omega$ and it is denoted by $\supp(\omega)$.
An $m$-ary fractional operation with finite support is a \emph{fractional polymorphism} if $\sum_{g \in \supp(\omega)}\omega(g)=1$ and for every $f \in \tau$ and tuples $a^1,\ldots,a^m \in C^{\ar(f)}$ it holds that \[\sum_{g \in \mathcal O^{(m)}_C}\omega(g)f^{\g}(g(a^1,\ldots,a^m))\leq \frac{1}{m}\sum_{i=1}^mf^{\g}(a^i)\] (where
$g$ is applied component-wise).
The set of all fractional polymorphisms of a valued structure $\g$ is denoted by $\fpol({\g})$. If $\omega$ is a fractional polymorphism of $\g$ we say that $\g$ is \emph{improved} by $\omega$.
\end{definition}


Let $S_m$ be the symmetric group on $\{1,\ldots,m\}$. An $m$-ary operation $g$ is \emph{fully symmetric} if for every permutation $\pi \in S_m$, we have \[g(x_1,\ldots,x_m)=g(x_{\pi(1)}, \ldots, x_{\pi(m)}).\] 
Note that every totally symmetric operation (see Definition \ref{def:totsym}) is also fully symmetric. 
A fractional operation is \emph{fully symmetric} (\emph{totally symmetric}, respectively) if every operation in its support is fully symmetric  (totally symmetric, respectively). 
The fractional operations in the next example are totally symmetric. 
\begin{example}\label{example:minmax}
    Let $D$ be a totally ordered set.
	Let $\min$ and $\max$ be, respectively, the binary operations giving the smallest and the largest among two arguments, respectively.
	The fractional operation $\omega_{\max}\colon\mathcal O^{(2)}_D \to [0,1]$, 	 \[\omega_{\max}(g):=\begin{cases}
	1 & \text{if }g=\max\\
	0 & \text{otherwise}
	\end{cases}\]
	is a binary totally symmetric fractional operation with finite support.
The fractional operation  $\omega_{\min}\colon\mathcal O^{(2)}_D \to [0,1]$, 	 is defined dually and it
	is a binary totally symmetric fractional operation with finite support, too.\demo
\end{example}

The next example shows a family of fractional operations that are fully symmetric but not totally symmetric.

\begin{example}\label{example:convex}
	Let $\avg^{(k)}\colon D^k \to D$ be the $k$-ary arithmetic average operation defined, for every $(x^1,\ldots,x^k) \in D^k$, by \[\avg(x^1,\ldots,x^k)=\frac{1}{k}\sum_{i=1}^kx^i.\]
	Let us define, for every $k \geq 2$, the fractional operation $\omega_{conv}^{(k)}\colon \mathcal O_D^{(k)} \to \Q_{\geq 0}$ such that \[\omega_{conv}^{(k)}(g)=\begin{cases}
	1 & \text{if } g =\avg^{(k)}\\
	0 & \text{otherwise.}
	\end{cases}\] 
	The fractional operations $\omega_{conv}^{(k)}$ are fully symmetric, but for $k \geq 3$  they are not totally symmetric,  because, e.g., $\{1,1,2\}=\{1,2,2\}$ but $\frac{1+1+2}{3}\neq\frac{1+2+2}{3}$.\demo
\end{example}

 \subsection{The Basic Linear Programming Relaxation}
 
 Every VCSP over a finite domain has a natural linear programming relaxation.
 Let $\g$ be a valued structure with finite domain $D$ and let $I$ be an instance of VCSP($\g$) with set of free variables $V_I=\{x_1,\ldots,x_d\}$, and objective function \[\phi_I(x_1,\ldots,x_d)=\sum_{ j \in J} f_j(x_1^j,\ldots,x_{n_j}^j) \text{,}\] with $ f_j \in\g,\; x^j=(x_1^j,\ldots,x_{n_j}^j)  \in V^{n_j}  \text{, for all } j \in J$ (the set $J$ is finite and indexing the cost functions that are summands of $\phi_I$).
Define the sets $W_1$, $W_2$, and $W$ of variables $\lambda_j(t)$, $\mu_{x_i}(a)$, for $j \in J$, $t \in D^{n_j}$, $x_i \in V$, and $a \in D$, as follows.  \begin{align*}
&W_1:=\{\lambda_{j}(t)\mid j \in J \text{ and }t=(t_1,\ldots,t_{n_j}) \in D^{n_j}\},\\
&W_2:=\{\mu_{x_i}(a) \mid x_i \in V \text{ and } a \in D\},\\
&W:=W_1\cup W_2.
\end{align*}
Then the \emph{basic linear programming, {BLP}, relaxation  associated to  $I$} (see \cite{ThapperZivny2012},
\cite{KolmogorovThapperZivny}, and references therein) is a linear program with variables  $W$ and  is defined as follows 
\[\label{blp} 
\blp(I, \g):=\min_{{\lambda_j(t)} \in \Q}{ \sum_{j \in J} \sum_{ t \in D^{n_j}} \lambda_{j}(t) f^{\g}_j(t)} \]
such that
\begin{align*}
 \sum_{t \in D^{n_j}: t_l=a}\lambda_j(t)=\mu_{x_l^j}(a) &\quad & \text{for all } j \in J \text{, } l \in \{1,\ldots,n_j\}\text{, } a \in D\\ 
\sum_{ a \in D}\mu_{x_i}(a)=1 &\quad &\text{for all } x_i \in V\\
\lambda_j(t)=0  &\quad & \text{for all } j \in J\text{, } t \notin \dom(f_j)\\
0\leq \lambda_j(t), \mu_{x_i}(a) \leq 1 &\quad & \text{for all } \lambda_j(t) \in W_1,\; \mu_{x_i}(a) \in W_2.
\end{align*}

If there is no feasible solution to this linear program, then $\blp(I,\g)=+\infty$. We say that the BLP relaxation \emph{solves} $I$ if $\blp(I,\g)=\min_{x \in D^d}\phi_I(x)$. We say that the BLP relaxation solves $\vcsp(\g)$ if it solves all instances of $\vcsp(\g)$.
For a given instance of a finite-domain VCSP, the corresponding BLP relaxation can be computed in polynomial time.  Therefore, if the VCSP for a valued structure $\g$ is solved by the BLP relaxation, then $\vcsp(\g)$ can be solved in polynomial time.

Let $\g$ be a valued structure that admits an efficient sampling algorithm.
We may solve $\vcsp(\g)$ by using the following algorithm that computes the BLP relaxation of the sample, and then solves the BLP relaxation.

\vskip \baselineskip

\begin{algorithm}[H] 
\SetAlgoNoLine
\KwIn{$I:=(V_I,\phi_I,u_I)$.}
\KwOut{accepts if there exists an assignment $h \colon V_I\to \dom(\g)$ such that $\phi_I(h(x_1,\ldots,x_{\lvert V_I\rvert }))\leq u_I.$}

$\Delta:=\text{Sampling}_{\g}(\lvert V_I\rvert)$\;
$\blp(I,\Delta)$\;
\eIf{$\blp(I,\Delta)\leq u_I$}{accept\;
		}{reject\;
	}            
      
\caption{Sampling $+$ BLP Algorithm}
\label{figure:algblp}
\end{algorithm}


\vskip \baselineskip
Note that Algorithm \ref{figure:algblp} runs in polynomial time in $\lvert V_I \rvert$, and that if it rejects, then  indeed the answer to $\vcsp(\g)$ is no, without further assumptions.



In the following we present a sufficient condition under which Algorithm \ref{figure:algblp} correctly solves $\vcsp(\g)$. Valued structures over a finite domain that can be solved by the BLP relaxation have been characterised by Kolmogorov, Thapper, and \v{Z}ivn\'{y}. 
\begin{theorem}[\cite{KolmogorovThapperZivny}, Theorem 1] \label{thm:thepwrofblpforfinitedom}
Let $\Delta$ be a valued structure with finite signature and finite domain. Then the following are equivalent:\begin{enumerate}
\item the $\blp$ relaxation solves $\vcsp(\Delta)$;
\item $\Delta$ has fully symmetric fractional polymorphisms of all arities.
\end{enumerate} 
\end{theorem}

One of the main results of this section, Theorem~\ref{thm:fullysymfpol}, states that  if $\g$ is improved by fully symmetric fractional operations of
all arities, then Algorithm \ref{figure:algblp} correctly solves $\vcsp(\g)$ in polynomial time. Note that there are examples of PLH valued structures which have fully symmetric fractional polymorphisms, but the computed sample
does not have fully symmetric fractional polymorphisms (Example \ref{ex:convexsubcase}).

\begin{definition}
Let $\Delta$ be a valued $\tau$-structure with domain $D$ and let $m \geq 1$. The \emph{multiset-structure} $\mathcal P^m(\Delta)$  \cite{ThapperZivny2012}  is the valued structure with domain $\left(\binom{D}{m}\right)$ i.e., the set of multisets of elements from $D$ of size $m$, and for every $k$-ary function symbol $f\in \tau$, and $\alpha_1, \ldots,\alpha_k \in \left(\binom{D}{m}\right)$ the function $f^{\mathcal P^m(\Delta)}$ is defined as follows
\[f^{\mathcal P^m(\Delta)}(\alpha_1,\ldots, \alpha_k):=\frac{1}{m}\min_{t^1,\ldots,t^k \in D^m: \{t^l\}=\alpha_l}\sum_{i=1}^{m}f^{\Delta}(t^1_i,\ldots,t^k_i).\]
(Here we denote by $\{t^l\}$ the multiset whose elements are the coordinates of $t^l$.)\end{definition}

\begin{lemma}[\cite{ThapperZivny2012}, Lemma 2.2] \label{lemma:TZfhom}
	Let $\Delta$ be a valued structure with finite domain, and $m\geq2$. Then $\mathcal P^m(\Delta)\to_f \Delta$ if, and only if, $\Delta$ has an $m$-ary fully symmetric fractional polymorphism.
\end{lemma}

\begin{lemma}\label{lemma:fhom}
	Let $\g$  be a valued structure (with finite signature $\tau$), and 
	let $\Delta$ be a valued structures with finite domain such that $\Delta \to_f \g$, and let $m\geq2$ be an integer. If  $\g$ has an $m$-ary fully symmetric fractional polymorphism, then $\mathcal{P}^m(\Delta) \to_f \g$. 
\end{lemma}

Lemma \ref{lemma:fhom} is a generalisation of Lemma \ref{lemma:TZfhom} to valued structures with arbitrary domain. However, while Lemma  \ref{lemma:TZfhom} follows directly from the definition of  $\mathcal P^m(\Delta)$, our proof of Lemma \ref{lemma:fhom} is inspired by the proof of \cite{BodMacpheTha}. In fact, Lemma \ref{lemma:fhom} is a generalisation of \cite{BodMacpheTha}, Lemma 2.4 to valued structures, and to the weaker assumption that the polymorphisms are fully symmetric rather than totally symmetric.

\begin{proof}[of Lemma \ref{lemma:fhom}]
 Let $C$ and $D$ be, respectively, the domain of $\g$ and the domain of $\Delta$, respectively. Let $\chi$ be a fractional homomorphism from $\Delta$ to $\g$ and let $\omega$ be an $m$-ary fully symmetric fractional polymorphism of $\g$. For every $g \in \supp(\omega)\subseteq \mathcal{O}_C^{(m)}$  and every $h \in \supp(\chi)\subseteq C^{D}$ we define \[(g\circ h)\colon \left(\binom{D}{m}\right )\to C\] by setting
	\[(g\circ h)(\alpha)=g(h(a^1),\ldots,h(a^m))\]for every  $\alpha=\{a^1,\ldots,a^m\}\in \left(\binom{D}{m}\right).$
	Observe that $(g\circ h)$ is well defined as $g$ is fully symmetric (the order of $h(a^1),\ldots,h(a^m)$ does not matter). We define the function $\omega'\colon C^{\left(\binom{D}{m}\right)}\to \Q_{\geq0}$ as  follows, for every $g \in C^{\left(\binom{D}{m}\right)}$, \[\omega'(g')=\sum_{g \in \supp(\omega)}\sum_{h \in \supp(\chi) : g \circ h=g'}\omega(g)\chi(h).\] 
	We claim that $\omega'$ is a fractional homomorphism from ${\mathcal{P}^m(\Delta})$ to $\g$. Indeed, the support $\supp(\omega')=\{(g\circ h)\mid g \in \supp(\omega),h \in \supp(\chi)\}$, is finite as the support of $\omega$ and the support of $\chi$ are finite. It also holds that
	\[\sum_{g' \in \supp(\omega')}\omega'(g')=\sum_{g \in \supp(\omega)} \omega(g) \sum_{h \in \supp(\chi)}\chi(h)=\sum_{g \in \supp(\omega)}\omega(g)=1.\]
Furthermore, for every $k$-ary  $f \in \tau$ and tuple $(\alpha_1,\ldots, \alpha_k)\in \left(\binom{D}{m}\right)^k$, with $\alpha_i=\{\alpha_i^1,\ldots,\alpha_i^m\}$, it holds that 
	\begin{align}\nonumber
			& \sum_{g' \in C^{\left(\binom{D}{m}\right)}}\omega'(g')f^{\g}(g'(\alpha_1,\ldots,\alpha_k))\\ \nonumber =&\hskip -.5\baselineskip \sum_{g \in \supp(\omega)}\sum_{h \in \supp(\chi)}\omega(g)\chi(h)f^{\g}(g(h(\alpha_1^1),\ldots,h(\alpha_1^m)),\ldots,g(h(\alpha_k^1),\ldots,h(\alpha_k^m)))\\ 
			\label{eq:frhom2}
			\leq & \hskip -.5\baselineskip\sum_{h \in \supp(\chi)}\chi(h)\left(\frac{1}{m}\sum_{i=1}^mf^{\g}\left(h(\alpha_1^{\pi_1(i)}),\ldots,h(\alpha_1^{\pi_1(i)})\right) \right)\\ \nonumber
			= & \frac{1}{m}\sum_{i=1}^m\sum_{h \in \supp(\chi)}\chi(h)f^{\g}\left(h(\alpha_1^{\pi_1(i)}),\ldots,h(\alpha_1^{\pi_1(i)})\right)\\\label{eq:frhom3}
			\leq & \frac{1}{m}\sum_{i=1}^mf^{\Delta} (\alpha^{\pi_1(i)}_1, \ldots, \alpha^{\pi_k(i)}_k) 
			\end{align}
	for every $\pi_1, \ldots,\pi_k \in S_m$. 
	 Inequality (\ref{eq:frhom2}) holds because $\omega$ is a fully symmetric fractional polymorphism of $\g$, and Inequality (\ref{eq:frhom3}) holds because $\chi$ is a fractional homomorphism from $\Delta$ to $\g$.
	Then, in particular it holds that \begin{align*}&\sum_{g' \in C^{\left(\binom{D}{m}\right)}}\omega'(g')f^{\g}(g'(\alpha_1,\ldots,\alpha_k)) \leq \frac{1}{m}\min_{t^1,\ldots,t^k \in D^m: \{t^l\}=\alpha_l}\sum_{i=1}^{m}f^{\Delta}(t_{i}^1, \ldots t_{i}^k)\\=& f^{\mathcal{P}^m(\Delta)}(\alpha_1,\ldots,\alpha_k).
	\end{align*}
\end{proof}

\begin{theorem}\label{thm:fullysymfpol}
	Let $\g$ be a valued structure with finite signature having fully symmetric fractional polymorphisms of all arities. If there exists an efficient sampling algorithm for $\g$, then Algorithm \ref{figure:algblp}  correctly solves $\vcsp(\g)$ (in polynomial time). 
\end{theorem}

\begin{proof}
	Let $I$ be an instance of $\vcsp(\g)$ with variables $V_I=\{x_1, \ldots,x_n\}$, objective function $\phi_I(x_1,\ldots,x_n)=\sum_{j \in J}\gamma_j(x^j)$ where $J$ is a finite set of indices, $\gamma_j \in \g$, and $x^j \in V_I^{\ar(f_j)}$, and threshold $u_I$. Let $\Delta$ be the finite-domain valued structure computed by the sampling algorithm for $\g$ on input $\lvert V_I\rvert$. 
	
	Let $C$ be the (possibly infinite) domain of $\g$ and $D $ the finite domain of $\Delta$. Note that if $\blp(I,\Delta) \nleq u_I$ (this also includes the case $\blp(I,\Delta)=+\infty$, i.e., the case that $I$ is not feasible) then $\inf_{d \in D^n}\phi_I(d)\nleq u_I$ which implies $\inf_{c \in C^n}\phi_I(c)\nleq u_I$ since $D$ was produced by the sampling algorithm for $\g$ on input $\lvert V_I\rvert$.  We may therefore safely reject. Otherwise, if $\blp(I,\Delta)\leq u_I$, then 
	$\inf_{\alpha \in \left(\binom{D}{m}\right)^n}\phi_I(\alpha)\leq \blp(I,\Delta)$. The proof of this last statement is contained in the first part of the proof of Theorem 3.2 in \cite{ThapperZivny2012}; we report it here for completeness.
	Let $(\lambda^\star, \mu^\star)$ be an optimal solution to $\blp(I, \Delta)$ and let $M$ be a positive integer such that $M\cdot \lambda^\star$, and $M\cdot \mu^\star$ are both integral. Let $\nu\colon V_I \to \left( \binom{D}{M}\right)$ be defined by mapping the variable $x_i$ to the multiset in which the elements are distributed accordingly to $\mu_{x_i}^\star$, i.e., for every $a \in D$ the number of occurrences of $a$ in $\nu(x_i)$ is equal to $M\mu_{x_i}^\star(a)$. Let $f_j$ be a $k$-ary function symbol in $\tau$ that occurs in a term $f_j(x^j)$ of the objective function $\phi_I$. 
	Now we write \[M\cdot \sum_{t \in D^k}\lambda^\star_j(t)f^{\Delta}_j(t)=f^{\Delta}_j(\alpha^1)+\cdots+f^{\Delta}_j(\alpha^M)\text{,}\] where the $\alpha^i \in D^k$ are such that $\lambda^\star_j(t)$-fractions are equal to $t$. Let us define $\alpha_l':=(\alpha^1_i,\ldots,\alpha^M_i)$ for $1\leq i \leq k$. We get 
	\begin{align*}&\sum_{ t \in D^{k}}\lambda^\star_j(t)f_j^{\Delta}(t)=\frac{1}{M}\sum_{i=1}^{M}f_j^{\Delta}(\alpha^i)=\frac{1}{M}\sum_{i=1}^{M}f_j^{\Delta}(\alpha_1^i, \ldots, \alpha_k^i)\\ \geq&  \frac{1}{M}\min_{t^1,\ldots,t^k \in D^M: \{t^l\}=\{\alpha_l'\}}\sum_{i=1}^{M}f_j^{\Delta}(t^1_i,\ldots,t^k_i)=f_j^{\mathcal P^M(\Delta)}(\alpha'_1,\ldots,\alpha'_k)\\=&f_j^{\mathcal P^M(\Delta)}(\nu(x))\text{, }\end{align*}
	where the last equality follows as the number of $a$'s in $\alpha_i'$ is \[M \cdot \sum_{ t \in D^{k}: t^i=a}\lambda^\star_j(t)=M\cdot \mu^\star_{x_i}(a).\] Then \begin{align*}
	\blp(I, \Delta)=&\sum_{ j \in J}\sum_{t \in D^{\ar(f_j)}}\lambda^\star_j(t)f_j^{\Delta}(t) = \sum_{ j \in J}\left (\sum_{t \in D^{\ar(f_j)}}\lambda^\star_j(t)f_j^{\Delta}(t) \right)\\ \geq&   \sum_{ j \in J}\left (f_j^{\mathcal P^M(\Delta)}(\nu(x))\right) \geq \inf_{\alpha \in \left(\binom{D}{m}\right)^n}\phi_I(\alpha).  \end{align*}
	
	Since we assumed $\blp(I, \Delta)\leq u_I$, we obtain $\inf_{\left(\binom{D}{m}\right)}\phi_I(\alpha)\leq u_I$. Moreover, since $\g$ has fully symmetric fractional polymorphisms of all arities, Lemma \ref{lemma:fhom} implies the existence of a fractional homomorphism $\omega\colon C^{\left(\binom{D}{M}\right)}\to \Q_{\geq0}$. From Proposition \ref{prop:frachom} it follows that \[\inf_{c \in C^n}\phi_I(c)\leq \inf_{\alpha \in \left(\binom{D}{m}\right)^n}\phi_I(\alpha)\leq \blp(I, \Delta)\leq u_I.\]
 \end{proof}

We  now  give examples of valued structures that satisfy the hypothesis of Theorem \ref{thm:fullysymfpol}. 
A set $S \subseteq \Q^n$ is said to be \emph{convex} if for any two points $x$, $y \in S$ every point between them is still in $S$, i.e., for any $\lambda \in [0,1] \cap \Q$, $\lambda x+(1-\lambda) y \in S$.
A function  $f \colon \Q^n \rightarrow \QQ$ is said to be \emph{convex} if for any two points $x$, $y \in \Q^n$ and for any $\lambda \in [0,1] \cap \Q$, it holds \[f(\lambda x+(1-\lambda) y)\leq \lambda  f(x)+(1-\lambda) f(y).\]

\begin{proposition}\label{prop:convexfop}
	Let $D\subseteq \Q$ be a convex set. 	
	Let $\g$ be a valued structure with domain $D$ such that every cost function in $\g$ is convex. Then, for every $k \geq 2$, the valued structure $\g$ is improved by the fully symmetric fractional operation $\omega_{conv}^{(k)}\colon \mathcal O_D^{(k)} \to \Q_{\geq 0}$ (see Example \ref{example:convex}) such that \[\omega_{conv}^{(k)}(g)=\begin{cases}	1 & \text{if } g =\avg^{(k)}\\	0 & \text{otherwise.}\end{cases}\]
\end{proposition}

\begin{proof}
	Let $f \colon D^n \to \QQ$ be a convex cost function.  Jensen's Inequality (cf.\ \cite{jensen1906}) implies that for all $k \geq 2$ and for all $x^1,\ldots,x^k \in D^n$ \[f\left(\frac{1}{k}\sum_{i=1}^kx^i\right)\leq \frac{1}{k}\sum_{i=1}^kf(x^i).\]
	Therefore, for every $k \geq 2$, the function $f$ is improved by the fully symmetric fractional operation $\omega_{conv}^{(k)}$.
\end{proof}

\begin{example}\label{ex:convexsubcase}Let $\g$ be a convex valued structure with domain $\Q$ that  admits an efficient sampling algorithm. Then, for all $k \geq 2$, the valued structure $\g$ is improved by the $k$-ary fully symmetric fractional operation $\omega_{conv}^{(k)}$ defined  in Example \ref{example:convex} (see Proposition \ref{prop:convexfop}), and therefore, by Theorem \ref{thm:fullysymfpol}, Algorithm \ref{figure:algblp}  solves $\vcsp(\g)$.
	Note that the finite-domain valued structure computed by the sampling algorithm might not have fully symmetric fractional polymorphisms of all arities. In fact, the fractional polymorphisms $\omega_{conv}^{(k)}$  are not even inherited by  valued finite  substructures of $\g$ whose domain contains more than one element.\demo
\end{example}

If the VCSP for a valued structure is solvable by  Algorithm \ref{figure:algblp}, then it is possible to find a solution  in polynomial  time, by applying the self-reduction algorithm (see \cite{KolmogorovThapperZivny}) to the instance $I$ of $\vcsp(\Delta)$. Observe that the assignment obtained by self-reduction has values in $D$ and therefore it is also an assignment with values in $C$.  

\begin{proposition}[\cite{KolmogorovThapperZivny}, Proposition 8]
 Let $\g$ be an arbitrary valued structure with finite domain. Let $I:=(V_I,\phi_I,u_I)$ be an instance of $\vcsp(\g)$. If the $\blp$ relaxation solves $\vcsp(\g)$, then an  assignment to the variables in $V_I$ with cost at most $u_I$ can be found in polynomial time.
\end{proposition}


We remark that Theorem \ref{thm:fullysymfpol} generalises the following known result (that we used in Section \ref{sect:csp-tract} to prove the polynomial-time complexity of the CSP for PLH relational structures preserved by $\max$).

\macphe*

More precisely, we extended Theorem \ref{macphe} to valued structures and, at the same time, to the weaker assumption of having fully symmetric polymorphisms of all arities rather than totally symmetric polymorphisms of all arities. 
The following example is adapted from \cite{Kun:2009:NLA:1536414.1536512} (Example 99), and it exhibits  a PL valued structure having fully symmetric polymorphisms of all arities but having no totally symmetric  polymorphism of arity $3$. 

\begin{example}
	Let us consider the PL valued structure $\g$ with domain $\Q$, signature $\{f_{+},f_{-}\}$, and such that $f_{+}^{\g},f_{-}^{\g}\colon \Q^3 \to \QQ$ are defined by \[f_{+}^{\g}(x_1,x_2,x_3):=\begin{cases}
	x_1+x_2+x_3 & \text{if } x_1+x_2+x_3\geq 1\\
	+\infty & \text{otherwise,}
	\end{cases}\]
	and
	\[f_{-}^{\g}(x_1,x_2,x_3):=\begin{cases}
	x_1+x_2+x_3 & \text{if } x_1+x_2+x_3\leq -1\\
	+\infty & \text{otherwise.}
	\end{cases}\]
	
	Clearly, the cost functions $f_{+}^{\g}$, and $f_{-}^{\g}$ are PL and it is easy to see that they are convex. As all cost functions in $\g$ are PL and convex, by Proposition \ref{prop:convexfop}, the valued structure $\g$ is improved by the fully symmetric fractional operations $\omega_{conv}^{(k)}$, for every $k\geq 2$, i.e., $\g$ has fully symmetric fractional polymorphisms of all arities. 
	We already observed that the fractional operations $\omega_{conv}^{(k)}$ are not totally symmetric for $k \geq 3$ (see Example \ref{example:convex}).
	
	Assume now that $\omega$ is a ternary totally symmetric fractional polymorphism of $\g$ and let $t\colon \Q^3 \to \Q$ be a totally symmetric operation in $\supp(\omega)$, then, in particular, $t$ is a polymorphism of $\feas(\g)$, i.e., $t$ preserves the feasibility relations   \begin{align*}
	&\dom(f_{+}^{\g})=\{(x_1,x_2,x_3) \in \Q^3 \mid x_1+x_2+x_3\geq 1 \}\text{, and}\\
	&\dom(f_{-}^{\g})=\{(x_1,x_2,x_3) \in \Q^3 \mid x_1+x_2+x_3\leq-1 \}.
	\end{align*}
	By the total symmetry of $t$, there exists $a \in \Q$ such that \begin{align*}&t(1,1,-1)=t(1,-1,1)=t(-1,1,1)\\=&t(-1,-1,1)=t(-1,1,-1)=t(1,-1,-1)=a.\end{align*}
	Observe that $(1,1,-1),(1,-1,1), (-1,1,1) \in \dom(f_{+}^{\g})$, then, by applying $t$ componentwise we get $(a,a,a) \in \dom(f_{+}^{\g})$, i.e., $a \geq \frac{1}{3}$;
	on the other hand, $(-1,-1,1),(-1,1,-1), (1,-1,-1) \in \dom(f_{-}^{\g})$, then, by applying $t$ componentwise we get $(a,a,a) \in \dom(f_{-}^{\g})$, i.e., $a \leq -\frac{1}{3}$, that is a contradiction.\demo
	
\end{example}

In the next section we will apply Theorem \ref{thm:fullysymfpol} to PLH valued structures.

\section{Polynomial-Time PLH VCSPs}
We apply the results of Sections \ref{sect:tract} and \ref{sect:totsym} to state the polynomial-time tractability of the VCSP for PLH valued structures with finite signature that are improved by fully symmetric operations of all arities.

\begin{proposition}\label{thm:PLHsym}
	Let $\g$ be a valued PLH constraint language with finite signature that is improved by fully symmetric fractional polymorphisms of all arities.
	Then $\vcsp(\g)$ can be solved in polynomial time.
\end{proposition}

\begin{proof}
The statement is an immediate consequence of Theorem \ref{thm:fullysymfpol}  and the fact that PLH valued structures with finite signature admit an efficient sampling algorithm (Corollary \ref{cor:PLHsampling}).
\end{proof}



\subsection{Submodular PLH Valued Structures}
\label{sect:submodular}
We have already mentioned the importance of submodularity in the introduction to this article.
In this section we  apply Proposition \ref{thm:PLHsym} to prove that the VCSP of submodular PLH valued structures with finite signature can be solved in polynomial time. Moreover, we show that submodularity defines a {\it maximally tractable} class of valued structures within the class of PLH valued structures. 

\begin{corollary}\label{cor:submodplh}
	Let $\g$ be a PLH valued structure with finite signature such that the cost functions in $\g $ are submodular. Then $\vcsp(\g)$ can be solved in polynomial time.
\end{corollary}

We already observed in Section \ref{sect:submodular} that a function  over a totally ordered set $D$ is submodular if, and only if, it is improved by the binary fully symmetric fractional operation $\omega_{\sub}\colon  \mathcal O_D^{2} \to \Q_{\geq 0}$ such that \[\omega_{\sub}(g)=\begin{cases}
\frac{1}{2} & \text{if } g=\max \\
\frac{1}{2} & \text{if } g=\min\\
0 & \text{otherwise.}
\end{cases}\] 
In fact, there is another equivalent characterisation of submodularity based on fractional polymorphisms.
For every $k\geq 2$ and every $i \in \{1,\ldots,k\}$ we define  $s^{(k)}_i\colon D^k \to D$ to be the operation returning the $i$-th smallest of its arguments with respect to the total order in $D$. Observe that for $k=2$ we get $s^{(2)}_1=\min$ and $s^{(2)}_2=\max$.
We define for every $k \geq 2$ the $k$-ary fractional operation $\omega_{\sub}^{(k)}\colon \mathcal O_D^{(k)} \to \Q_{\geq0}$ having support \[\supp(\omega_{\sub}^{(k)})=\{s^{(k)}_i\ \mid 1 \leq i \leq k\}\] by setting 
\[\omega_{\sub}^{(k)}(g):= \begin{cases}
\frac{1}{k} & \text{if } g \in \supp(\omega_{\sub}^{(k)})\\
0 & \text{otherwise.}
\end{cases}\]
The operations $s^{(k)}_i(x^1,\ldots,x^k)$ are fully symmetric for all $i \in \{1,\ldots,k\}$ and all $k \geq 2$, therefore the fractional operations $\omega_{\sub}^{(k)}$ are fully symmetric for all $k \geq 2$. It is an easy observation that, for $k=2$, the fractional operation $\omega_{\sub}^{(2)}$ is exactly the fractional operation $\omega_{\sub}$ characterising submodular functions.

\begin{proposition}\label{prop:gensub}
	Let $D$ be a totally ordered set and let $f$ be a submodular function over $D$. Then for all $k \geq 2$, the fractional operation $\omega_{\sub}^{(k)}$ improves the function $f$. 
\end{proposition}

\begin{proof}
	Clearly, \[\sum_{g \in \supp(\omega_{\sub}^{(k)})}\omega_{\sub}^{(k)}(g)=1.\]
 We want to prove that for all $k \geq2$ and for all $x^1, \ldots, x^k \in D^n$ it holds that \begin{equation}\label{eq:gensub}\frac{1}{k}\sum_{i=1}^{k}f(s^{(k)}_i(x^1, \ldots, x^k))\leq \frac{1}{k}\left(f(x^1)+\cdots +f(x^k)\right).\end{equation}
By using the submodularity of $f$ we can write  \begin{align}\nonumber &f(x^1)+\cdots+ f(x^k)=\frac{1}{k-1}\sum_{1\leq i< j \leq k} \left(f(x^i)+f(x^j)\right)   \\ \nonumber \geq &\frac{1}{k-1} \sum_{1\leq i< j \leq k} \left(f(\min(x^i,x^j))+f(\max(x^i,x^j))\right)  \\ \nonumber = &\sum_{1\leq i< j \leq k} \frac{f(\min(x_1^i,x_1^j),\ldots,\min(x_n^i,x_n^j))+f(\max(x^i_1,x^j_1), \ldots,\max(x_n^i,x_n^j))}{k-1}\\
\label{eq:gensub4} \geq & \sum_{i=1}^{k}f\left(s^{(k)}_i(x^1_1,\ldots,x^k_1),\ldots,s^{(k)}_i(x^1_n,\ldots,x^k_n)\right),
 \end{align}
 from which Inequality (\ref{eq:gensub}) follows.
 We prove Inequality (\ref{eq:gensub4}) by induction on $n$.
	Observe that for every coordinate $1\leq l \leq n$, the following equality between  multisets holds:	\begin{align} &\left\{\min(x_l^i,x_l^j) \mid 1 \leq i <j \leq k\right\}\cup \left\{\max(x_l^i,x_l^j)\mid 1 \leq i <j \leq k\right\} \nonumber \\ \label{eq:gensub3}  = &\; \big\{\underbrace{s^{(k)}_1(x^1_l,\ldots,x^k_l)}_{k-1 \text{ occurrences}}, \underbrace{s^{(k)}_2(x^1_l,\ldots,x^k_l)}_{k-1\text{ occurrences}},\ldots, \underbrace{s^{(k)}_{k}(x^1_l,\ldots,x^k_l)}_{k-1 \text{ occurrences}}\big\}.\end{align}
	If $f$ has arity $n=1$, then Inequality (\ref{eq:gensub4}) immediately follows from Equality (\ref{eq:gensub3}).
	Let $n \geq 2$, assume that Inequality (\ref{eq:gensub4}) is true for submodular functions of arity at most $n-1$,  and let us prove it for submodular functions of arity $n$.
	From Equality (\ref{eq:gensub3}) and from the inductive hypothesis it follows that there exist $(k-1)$-many permutations ${\pi}_{1},\ldots, {\pi}_{k-1} \in S_k$
	such that
	\begin{align}\nonumber&\hskip -.5\baselineskip\sum_{1\leq i< j \leq k}\hskip -.3\baselineskip \left(f(\min(x_1^i,x_1^j),\ldots,\min(x_n^i,x_n^j))+f(\max(x^i_1,x^j_1), \ldots,\max(x_n^i,x_n^j))\right)\\\label{ineq:above1}\geq & \sum_{p=1}^{k-1} \sum_{i=1}^{k} f\left(s^{(k)}_i(x^1_1,\ldots,x^k_1),\ldots,s^{(k)}_i(x^1_{n-1},\ldots,x^k_{n-1}),s^{(k)}_{\pi_p(i)}(x^1_n,\ldots,x^k_n)\right).\end{align}
We claim that for every $p \in \{1,\ldots,k-1\}$ it holds that \begin{align}
\nonumber&\sum_{i=1}^{k}f\left(s^{(k)}_i(x^1_1,\ldots,x^k_1),\ldots,s^{(k)}_i(x^1_{n-1},\ldots,x^k_{n-1}),s^{(k)}_{\pi_p(i)}(x^1_n,\ldots,x^k_n)\right)\\ \label{ineq:above} \geq& \sum_{i=1}^{k}f\left(s^{(k)}_i(x^1_1,\ldots,x^k_1),\ldots,s^{(k)}_i(x^1_{n-1},\ldots,x^k_{n-1}),s^{(k)}_{i}(x^1_n,\ldots,x^k_n)\right). 
\end{align} To prove Inequality (\ref{ineq:above}), let $j:=\max\left\{i \in\{1,\ldots,k\} \mid \pi_p(i)\neq i \right\}$. Then there exists $l \in \{1,\ldots,j-1\}$ such that $\pi_p(l)=j$. By the submodularity of $f$ we have that \begin{align*}
	&f\left(s^{(k)}_j(x^1_1,\ldots,x^k_1),\ldots,s^{(k)}_j(x^1_{n-1},\ldots,x^k_{n-1}),s^{(k)}_{\pi_p(j)}(x^1_n,\ldots,x^k_n)\right)\\+\;&f\left(s^{(k)}_l(x^1_1,\ldots,x^k_1),\ldots,s^{(k)}_l(x^1_{n-1},\ldots,x^k_{n-1}),s^{(k)}_{j}(x^1_n,\ldots,x^k_n)\right)\\ \geq\;& f\left(s^{(k)}_j(x^1_1,\ldots,x^k_1),\ldots,s^{(k)}_j(x^1_{n-1},\ldots,x^k_{n-1}),s^{(k)}_{j}(x^1_n,\ldots,x^k_n)\right)\\+\;&f\left(s^{(k)}_l(x^1_1,\ldots,x^k_1),\ldots,s^{(k)}_l(x^1_{n-1},\ldots,x^k_{n-1}),s^{(k)}_{\pi_p(j)}(x^1_n,\ldots,x^k_n)\right).
	\end{align*}
	After this step  \begin{align*}
	&\sum_{i=1}^{k} f\left(s^{(k)}_i(x^1_1,\ldots,x^k_1),\ldots,s^{(k)}_i(x^1_{n-1},\ldots,x^k_{n-1}),s^{(k)}_{\pi_p(i)}(x^1_n,\ldots,x^k_n)\right)\\ \geq \;& \sum_{i=j}^{k}f\left(s^{(k)}_i(x^1_1,\ldots,x^k_1),\ldots,s^{(k)}_i(x^1_{n-1},\ldots,x^k_{n-1}),s^{(k)}_{i}(x^1_n,\ldots,x^k_n)\right)\\ &+ \sum_{i=1}^{j-1}f\left(s^{(k)}_i(x^1_1,\ldots,x^k_1),\ldots,s^{(k)}_i(x^1_{n-1},\ldots,x^k_{n-1}),s^{(k)}_{\pi'_p(i)}(x^1_n,\ldots,x^k_n)\right)\end{align*}
	where $\pi'_{p} \in S_k$ is the permutation  defined by \[\pi'_p(i)=\begin{cases}
	\pi_p(j) & \text{if } i=l\\
	\pi_p(i) & \text{otherwise.}
	\end{cases}\] 
	By reiterating the procedure for at most  $j-1 \leq k$ times for every $p \in \{1,\ldots,k-1\}$, we get the claim. Then, using Inequality (\ref{ineq:above}), we can rewrite Inequality (\ref{ineq:above1}) as \begin{align*}&\hskip -.4\baselineskip\sum_{1\leq i< j \leq k}\hskip -.4\baselineskip \left(f(\min(x_1^i,x_1^j),\ldots,\min(x_n^i,x_n^j))+f(\max(x^i_1,x^j_1), \ldots,\max(x_n^i,x_n^j))\right)\\\geq &\;\sum_{p=1}^{k-1} \sum_{i=1}^{k}f\left(s^{(k)}_i(x^1_1,\ldots,x^k_1),\ldots,s^{(k)}_i(x^1_{n-1},\ldots,x^k_{n-1}),s^{(k)}_{i}(x^1_n,\ldots,x^k_n)\right)\\=&\;\sum_{p=1}^{k-1}\sum_{i=1}^{k}f\left(s^{(k)}_i(x^1,\ldots,x^k)\right)=(k-1)\sum_{i=1}^{k}f\left(s^{(k)}_i(x^1,\ldots,x^k)\right),\end{align*}
	that is, Inequality (\ref{eq:gensub4}) holds and this concludes the proof.
\end{proof}

The next corollary immediately follows from the full symmetry of fractional operations $\omega_{\sub}^{(k)}$ and  Proposition \ref{prop:gensub}.

\begin{corollary}\label{cor:gensub}
	Let $\g$ be a valued PL (or PLH) submodular constraint language. Then $\g$ has fully symmetric fractional polymorphisms of all arities.
\end{corollary}

\begin{proof}[of Corollary \ref{cor:submodplh}]
The proof follows directly from Proposition \ref{thm:PLHsym}, since every valued submodular PLH constraint language has fully symmetric fractional polymorphysms of all arities  (Corollary \ref{cor:gensub}).
\end{proof}

We show that  submodular PLH valued structures are {\it maximal tractable}  within the class of PLH valued structures.
Let $\g$ be a valued structure with signature $\tau$.
A \emph{valued reduct} $\g'$ of  $\g$  is a valued structure with domain $\dom(\g)$ and such that the signature $\tau'$ can be obtained from $\tau$ by dropping some of the function symbols. The interpretation of function symbols from $\tau'$ in $\g'$ is as in $\g$. In this case we also say that $\g$ is an \emph{expansion} of $\Delta$.
A valued reduct of a valued structure is said \emph{finite} if its signature is finite.

\begin{definition}
	Let $\mathcal V$ be a class of valued structures with fixed domain $D$ and let $\g$ be a valued structure in $\mathcal V$. We say that $\g$ is \emph{maximal tractable within} $\mathcal V$ if \begin{itemize}
		\item $\vcsp(\g')$ is polynomial time solvable for every valued finite reduct $\g'$ of $\g$; and
		\item for every valued structure $\Delta$ in $\mathcal V$ that is an expansion of $\g$, there exists a valued finite reduct $\Delta'$ of $\Delta$ such that $\vcsp(\Delta')$ is NP-hard.
	\end{itemize}
\end{definition}

We will make use of the following result.

\begin{theorem}[\cite{cohen2006complexity}, Theorem 6.7]\label{ccjk2}
	Let $D$ be a finite totally ordered set. 
	Then the valued structure containing all submodular cost functions over $D$ is maximal tractable within the class of all valued structures with domain $D$. 
\end{theorem}

We show that the class of submodular PLH valued structures is maximal tractable within the class of PLH valued structures.

\begin{definition}
	Given a finite domain $D \subset \Q$ and a partial function $f\colon D^n \rightarrow \mathbb Q$ we define the \emph{canonical extension} of $f$ as $\hat{f}\colon \Q^n \rightarrow \Q$, by \[\hat f(x)=\begin{cases}
	f( x) & x \in D^n\\
	+\infty & \text{otherwise}.
	\end{cases}\]
\end{definition}

Note that the canonical extension of a submodular function over a finite domain is submodular and PLH. 

\begin{theorem}\label{th:mt}
	The valued structure containing all submodular PLH cost functions is maximal tractable within the class of PLH valued structures. 
\end{theorem}

\begin{proof}
	Polynomial-time tractability of the VCSP for submodular PLH valued structures with finite signature has been stated in Theorem~\ref{cor:submodplh}. 
	
	Let $f^{\Q}$ be an $m$-ary PLH cost function over ${\mathbb Q}$ that is not submodular, i.e., 
	there exists a couple of points, $ a:=(a_1,\ldots,a_m)$, $ b:=(b_1,\ldots,b_m) \in \Q^m$ such that \[f^{\Q}( a)+f^{\Q}( b)<f^{\Q}(\min( a,  b))+f^{\Q}(\max( a,  b)).\]
	Let \[D:=\{a_1,\ldots,a_m,b_1,\ldots,b_m\}\subset \Q,\] 
	and let $\Delta$ be the valued structure with domain $D$ and such that its signature $\tau$ contains a function symbol for every submodular cost functions on $D$. 
	Notice that the restriction $f^{\Q}|_D$ is not submodular, for our choice of $D$.
	Therefore, by Theorem \ref{ccjk2}, there exists a valued structure $\Delta'$ having domain $D$ and signature $\tau' \cup \{f\}$, where $\tau' \subseteq \tau$ is finite  and the cost function $f^{\Delta'}=f^{\Q}_|D$, such that $\vcsp(\Delta')$ NP-hard.
	
	We define $\g'$ to be the (submodular PLH) valued structure with domain $D$, signature $\tau' \cup \{f, \chi_D\}$, and such that the interpretation of functions symbols in the signature is as it follows:\begin{itemize}
		\item for every $g \in \tau'$, the cost function $g^{\g'}$ is the canonical extension of $g^{\Delta'}$,
		\item the cost function $f^{\g'}$ is $f^{\Q}$, and
		\item the unary cost function $\chi_D^{\g'}\colon \Q \to \QQ$ is defined, for every $x \in \Q$ as \[\chi_D^{\g'}(x)=\begin{cases}
		0 & \text{if } x \in D\\
		+\infty & \text{if } x \in \Q \setminus D.\\ 
		\end{cases}\]
		Observe that $\chi_D^{\g'}$ is submodular and PLH.
	\end{itemize}
	The valued structure $\g'$ is an extension of an expansion (or an expansion of an extension) of $\Delta'$. We claim that $\vcsp(\g')$ in NP-hard. Indeed, for every instance  $I:=(V, \phi_I,u )$ of $\vcsp(\Delta')$, with $V:=\{v_1,\ldots,v_n\}$, we define the instance $J:=(V,\phi_J,u)$ of $\vcsp(\g')$ such that 
	\[\phi_J(v_1,\ldots,v_n):=\phi_I(v_1,\ldots,v_n)+\sum_{i=1}^{n}\chi_D(v_i).\]
	Because of the terms involving $\chi_D$, 
	an assignment $h \colon V\to \Q$ is such that $\phi_J^{\g'}(h(v_1),\ldots,h(v_n))$
	is smaller than $+\infty$ if, and only if, $h(v_i)\in D$ for all $v_i \in V$. In this case, \[\phi_J^{\g'}(h(v_1),\ldots,h(v_n))=\phi_I^{\Delta'}(h(v_1),\ldots,h(v_n))\,\] and therefore, deciding whether there exists an assignment $h \colon V \to \Q$ such that \[\phi_J^{\g'}(h(v_1),\ldots,h(v_n))\leq u\] is equivalent to decide whether there exists  an assignment $h' \colon V \to D$ such that \[\phi_I^{\Delta'}(h'(v_1),\ldots,h'(v_n))\leq u.\] 
	Since $J$ is computable in polynomial-time from $I$, the NP-hardness of the problem $\vcsp(\g')$ follows from the NP-hardness of 
	$\vcsp(\Delta')$. 
\end{proof}

\subsection{Convex PLH Valued Structures}

We have already seen that  if $\g$ is a valued structure such that all cost functions in $\g$ are convex, then $\g$ has fully symmetric fractional polymorphisms of all arities (see Proposition \ref{prop:convexfop}). Therefore, the next corollary follows directly from Proposition \ref{prop:convexfop} and Corollary \ref{thm:PLHsym}.

\begin{corollary}\label{cor:PLHconv}
	Let $\g$ be a PLH valued structure  with finite signature such that the cost functions in $\g$ are convex. Then $\vcsp(\g)$ can be solved in polynomial time.
\end{corollary} 

\subsection{Componentwise Increasing PLH Valued Structures}

Let $f\colon \Q^n \to \Q\cup\{+\infty\}$ be an $n$-ary  function. We say that \begin{itemize}
	\item $f$ is \emph{componentwise increasing} if \[f(x_1,\ldots,x_{i-1},y_i, x_{i+1},\ldots,x_n)\leq f(x_1,\ldots,x_{i-1},z_i, x_{i+1},\ldots,x_n)\text{,}\]
	for every $y_i <z_i $ , $1\leq i \leq k$, and $x_1,\ldots,x_{i-1},x_{i+1},\ldots, x_n \in \Q$;
	\item $f$ is \emph{componentwise decreasing} if \[f(x_1,\ldots,x_{i-1},y_i, x_{i+1},\ldots,x_n)\geq f(x_1,\ldots,x_{i-1},z_i, x_{i+1},\ldots,x_n)\text{,}\]
	for every $y_i <z_i $, $1\leq i \leq k$, and $x_1,\ldots,x_{i-1},x_{i+1},\ldots, x_n \in \Q$;

\end{itemize}

In \cite{cohen2006complexity}, componentwise increasing functions, and componentwise decreasing functions are respectively referred to as \emph{monotone} functions and \emph{antitone} functions.

Let us define, for every $k \geq 2$, the fractional operations $\omega_{\min}^{(k)}\colon \mathcal O_D^{(k)} \to \Q_{\geq 0}$, and  $\omega_{\max}^{(k)}\colon \mathcal O_D^{(k)} \to \Q_{\geq 0}$ by setting, respectively, \[\omega_{\min}^{(k)}(g):=\begin{cases}
1 & \text{if } g =\min^{(k)}\\
0 & \text{otherwise,}
\end{cases}\]  
and $\omega_{\max}^{(k)}(g)$ is dually defined.
The fractional operations $\omega_{\min}^{(k)}$ and $\omega_{\max}^{(k)}$ are fully symmetric (in fact they are totally symmetric).

A valued structure $\g$ is said to be \emph{componentwise increasing} (\emph{componentwise decreasing}, respectively) if all the cost functions in $\g$ are componentwise increasing (componentwise decreasing, respectively).

\begin{corollary}\label{cor:PLHincr}
	Let $\g$ be a componentwise increasing 
	 PLH valued structure  with finite signature. Then $\vcsp(\g)$ can be solved in polynomial time.
\end{corollary}

\begin{proof}
The proof follows from Proposition \ref{thm:PLHsym}, 
since for every $k\geq 2$ the fully symmetric fractional operation $\omega^{(k)}_{\min}$ 
improves $\g$ as the next lemma shows.
\end{proof}

\begin{lemma}\label{incr}
	A function $f \colon \Q^n \to \Q\cup \{+\infty\}$ is componentwise increasing 
	if, and only if, $f$ is improved by $\omega^{(k)}_{\min}$ 
	, for every $k\geq 2$. \end{lemma}
\begin{proof}Let $k \geq 2$ and let us first consider $f$ componentwise increasing with arity $1$. Let us consider $x_1,  \ldots,x_k \in \Q$ and let us assume without loss of generality that $x_1=\min(x_1, \ldots, x_k)$. Then we have that \[f(\min(x_1,\ldots, x_k))=f(x_1)=\frac{\overbrace{f(x_1)+\cdots+f(x_1)}^{k \text{ times}}}{k}\leq\frac{f(x_1)+\cdots+f(x_k)}{k}\text{,}\]
	i.e., $\omega_{\min}^{(k)}$ is a fractional polymorphism of $f$. Assume now that, for $n\geq 2$, every $(n-1)$-ary componentwise increasing function is improved by $\omega_{\min}^{(k)}$, and let us prove that for every $n \in \mathbb N$, an $n$-ary componentwise increasing function $f$ is improved by $\omega_{\min}^{(k)}$. Let us fix $x^1,\ldots,x^k\in \Q^{n}$.  The restricted function \[f(\cdot, \min(x^1_n,\ldots,x^k_n) ) \colon \Q^{n-1}\to \QQ\text{,}\]
	which maps every $(z_1,\ldots,z_{n-1}) \in \Q^{n-1}$ to $f(z_1,\ldots,z_{n-1}, \min(x^1_n,\ldots,x^k_n))$, is clearly componentwise increasing and therefore, by the inductive hypothesis, it is improved by $\omega_{\min}^{(k)}$, that is,  \begin{align*}&f(\min(x^1_1,\ldots,x^k_1),\ldots, \min(x^1_{n-1},\ldots,x^k_{n-1}), \min(x^1_n,\ldots,x^k_n))\\ \leq\;&\frac{f(x^1_1,\ldots,x^1_{n-1},\min(x^1_n,\ldots,x^k_n))+\cdots+f(x^k_1,\ldots,x^k_{n-1},\min(x^1_n,\ldots,x^k_n))}{k}.\end{align*}

	Again by the fact that $f$ is componentwise increasing we get that the right-hand side of the last inequality is  \[\leq \; \frac{1}{k}(f(x^1_1,\ldots,x^1_{n-1},x^1_n)+\cdots+f(x^k_1,\ldots,x^k_{n-1},x^k_n)),\] i.e., $f$ is improved by $\omega_{\min}^{(k)}$. 
	
Conversely, if $f \colon \Q^n \to \Q\cup \{+\infty\}$ is improved by $\omega^{(k)}_{\min}$ for all $k\geq 2$, then in particular it is improved by $\omega_{\min}=\omega_{\min}^{(2)}$. For all $i \in \{1,\ldots,n\}$, and $x_1,\ldots,x_{i-1}, y_i, z_i, x_{i+1},\ldots, x_n \in \Q$  such that $y_i <z_i$ it holds that \begin{align*}&f(x_1,\ldots,x_{i-1},y_i,x_{i+1},\ldots,x_n) \\ \leq \;&\frac{f(x_1,\ldots,x_{i-1},y_i,x_{i+1},\ldots,x_n)}{2}+\frac{f(x_1,\ldots x_{i-1},z_i,x_{i+1},\ldots,x_n)}{2}.\end{align*} It follows that 
	\[f(x_1,\ldots,x_{i-1},y_i,x_{i+1},\ldots,x_n)\leq f(x_1,\ldots x_{i-1},z_i,x_{i+1},\ldots,x_n),\]
	i.e., $f$ is componentwise increasing.
 \end{proof}

Lemma \ref{incr} and Corollary \ref{cor:PLHincr} may be shown to have a dual form that holds in the case of componentwise decreasing cost functions.

The valued structure containing all componentwise decreasing PLH functions (all componentwise
increasing PLH cost functions, respectively) is maximal tractable within the class of PLH valued structures. The proof is similar to the proof of the maximal tractability of submodular PLH valued structures (Theorem \ref{th:mt}) and uses a result by Cohen, Cooper, Jeavons, and Krokhin stating the maximal tractability of  componentwise decreasing finite-domain valued structures (see \cite{cohen2006complexity}, Theorem 6.15).

\section{Conclusion and Outlook}
\label{sect:conclusion}

We have presented two main results: an efficient sampling algorithm for PLH valued structures with finite signature, and an algebraic condition for the VCSP of infinite-domain valued structures that admit an efficient sampling algorithm to be solved in polynomial time by  the  BLP relaxation. By combining these two results we obtained a polynomial-time algorithm for PLH valued structures over the rationals that are improved by fully symmetric fractional operations of all arities. Finally, we showed that this polynomial-time algorithm solves the VCSP for PLH cost functions that are submodular, convex, or piecewise increasing.  
In fact, our algorithm not only decides whether there exists a solution
of cost at most $u_I$, but it can also be adapted to efficiently compute 
the infimum of the cost of all solutions (which might be $-\infty$),
and decides whether the infimum is attained. 
The modification is straightforward observing that the computed sample  does not depend on the threshold $u_I$.

We also showed that submodular PLH cost functions are \emph{maximally tractable} within the class of PLH cost functions. 
Such maximal tractability results are of particular importance for the more ambitious goal of classifying the complexity of the VCSP for all classes of
PLH cost functions: to prove a complexity dichotomy it suffices to identify \emph{all} maximally tractable classes. 
In the same direction, the study of the algebraic properties that make a PLH valued structure NP-hard, namely an  algebraic theory of hardness, is an interesting field of future research.
Another challenge is to understand which algebraic properties of fractional polymorphisms
of an infinite-domain valued structure are necessary the polynomial-time
tractability of the VCSP under the assumption that the valued structure admits an
efficient sampling algorithm.
Finally, we would like to extend our tractability
result to the class of all submodular \emph{piecewise linear} VCSPs. We believe that submodular 
piecewise linear VCSPs are in P, too.   
But note that already the 
structure $(\Q;0,S,D)$ 
where $S := \{(x,y) \mid y = x+1\}$
and $D := \{(x,y) \mid y = 2x\}$ 
(which has both min and max as polymorphisms)
does \emph{not} admit
an efficient sampling algorithm (it is easy to see that for every $d \in {\mathbb N}$ every sample computed on input $d$ must have exponentially many vertices in $d$), so a different approach than the approach in this article is needed.

\section*{Acknowledgement}

We would like to thank Jakub Opr\v{s}al for helpful discussions and comments on a draft of this article.

\bibliographystyle{abbrv}
\bibliography{local.bib}


\end{document}